%% file: main.tex
\documentclass[12pt]{article}
\usepackage{authblk}
\RequirePackage{amsthm,amsmath}
\RequirePackage[numbers]{natbib}
\RequirePackage[colorlinks,citecolor=blue,urlcolor=blue]{hyperref}
\usepackage{graphicx}
\usepackage[utf8]{inputenc} %
\usepackage{amsfonts}       
\usepackage{nicefrac}       
\usepackage{microtype}      
\usepackage{lipsum}
\usepackage{subfig} 
\usepackage{caption}
\usepackage{multirow}
\usepackage{multicol}
\usepackage{setspace}
\usepackage{footnote}
\usepackage{diagbox} 
\usepackage{float}
\usepackage{threeparttable}
\usepackage{booktabs} 
\usepackage{amsmath}
\usepackage{amsthm}
\usepackage{amssymb}
\usepackage{color}
\usepackage{mathtools}

\usepackage{enumitem}

\newtheorem{lemma}{Lemma}
\theoremstyle{plain}
\newtheorem{theorem}{Theorem}
\newtheorem{prop}{Proposition}
\newtheorem{remark}{Remark}
\newtheorem{corollary}{Corollary}
\newtheorem{assumption}{Assumption}

\newcommand{\R}{\mathbb R}
\def\tr{\text{tr}}
\def\t{\intercal}
\def\F{\mathcal{F}}
\def\la {\langle}
\def\ra {\rangle}
\def\ba{\boldsymbol{a}}
\def\bv{\boldsymbol{v}}
\def\bx{\boldsymbol{x}}

\def\bmu{\boldsymbol{\mu}}
\def\I{\mathbb{I}}
\def\Pihat{\hat{\Pi}}
\def\Hhat{\hat{H}}
\def\Zhat{\hat{Z}}
\def\bvhat{\hat{\bv}}

\def\bu{\boldsymbol{u}}

\def\by{\boldsymbol{y}}

\def\bdot{\boldsymbol{\cdot}}
\usepackage[toc,page]{appendix}
\def\by{\boldsymbol{y}}
\def\bx{\boldsymbol{x}}
\def\be{\boldsymbol{e}}

\def\Sigmahat{\hat{\Sigma}}

\author[1]{Lin Xiao}
\author[1]{Luo Xiao}
\affil[1,2]{North Carolina State University}
\date{}

\begin{document}

\title{Sparse and Integrative Principal Component Analysis for Multiview Data}

\maketitle

 \begin{abstract}
     We consider dimension reduction of multiview data, which are emerging in scientific studies. 
     Formulating multiview data as multivariate data with block structures corresponding to the different views, or views of data,
     we estimate top eigenvectors from multiview data that have two-fold sparsity, elementwise sparsity and blockwise sparsity. We propose a Fantope-based optimization criterion with multiple penalties to enforce the desired sparsity patterns and a denoising step is employed to handle potential presence of heteroskedastic noise across different data views. An alternating direction method of multipliers (ADMM) algorithm is used for optimization. We derive the
     $\ell_2$ convergence of the estimated top eigenvectors and establish their sparsity and support recovery properties. 
     Numerical studies are used to illustrate the proposed method.
\end{abstract}


\maketitle

\input{sections/Introduction}

\input{sections/Methods}

\input{sections/Theoretical_properties/General_theoretical_results}

\input{sections/Theoretical_properties/Spiked_covariance_model_with_strong_signals}

\input{sections/Theoretical_properties/Spiked_covariance_model_with_weak_signals}

\input{sections/Simulation}

\input{sections/Real_data}

\input{sections/Discussion}

\newpage 
\appendix
\section{Appendix: technical proofs}
\subsection{Proofs for Section \ref{sec:general}}

\input{sections/Appendix/Technical_proofs/Propositions/Proposition2}

\input{sections/Appendix/Technical_proofs/Propositions/Proposition1}

\input{sections/Appendix/Technical_proofs/Propositions/Proposition3}

\input{sections/Appendix/Technical_proofs/Propositions/Proposition4}

\subsection{Proofs for Section \ref{sec:strong}}
\input{sections/Appendix/Technical_proofs/Theorems/Theorem3}

\input{sections/Appendix/Technical_proofs/Theorems/Theorem4}

\subsection{Proofs for Section \ref{sec:weak}}

\input{sections/Appendix/Technical_proofs/Theorems/Theorem5}

\input{sections/Appendix/Technical_proofs/Theorems/Theorem6}
\input{sections/Appendix/Technical_proofs/Theorems/Theorem7}

\input{sections/Appendix/Technical_proofs/Theorems/Theorem8}

\subsection{Auxiliary lemmas}
\input{sections/Appendix/Technical_proofs/Lemmas}


\bibliography{reference}
\bibliographystyle{apalike}

\end{document}

%% file: sections/Introduction.tex
\section{Introduction}\label{intro}

Multiview data which refers to multiple sets of features generated on the same set of subjects has become increasingly prevalent in a variety of fields. In biomedical research, multiple types of genomics
data measured by different technologies on disparate platforms have been produced to gain insight into cancer genetics and molecular biology (\cite{lock2013bayesian}). A typical example is The Cancer
Genome Atlas (TCGA, \url{https://www.cancer.gov/about-nci/organization/ccg/research/structural-genomics/tcga}) which consists of measurements of RNA gene expression (GE), DNA methylation (ME), miRNA expression (miRNA), and reverse phase protein array (RPPA) on a
common set of 348 breast cancer tumor samples and each dataset is referred to as a data
view. In the area of multimedia analytics, multiview data can be generated by representing
images/texts/videos from disparate perspectives. For instance, the UCI multiple features
dataset (\url{https://archive.ics.uci.edu/ml/datasets/Multiple+Features}) contains 2000 images for handwritten numerals from ‘0’ to ‘9’, where each image is presented
in terms of six categories of visual features, including Fourier coefficients, profile correlations, Karhunen-Loeve coefficients, pixel values, Zernike moments and morphological features. 

One major focus of multiview data analysis is to uncover the underlying associations across different data views. For example, the four data views in TCGA might represent different biological components but they are also considered biologically related since they are generated from the same set of subjects. Joint analysis of all data views makes it possible to identify this type of association. Thus, there has been growing interest and demands in 
effectively combining the strength of each data view
through an integrative statistical learning paradigm, e.g., integrative clustering (\cite{kumar2011co},\cite{ye2018new},\cite{gao2020clusterings},  \cite{wang2021integrative}), integrative classification (\cite{shu2019multi},\cite{zhang2021joint}) and integrative regression (\cite{li2019integrative}).

One fundamental tool in the joint analysis of multiview data is dimension reduction, through which  meaningful decomposition of variations across different data views is identified. Canonical correlation analysis (CCA, \cite{hotelling1936relations}) is a classic method to capture joint variation between two data views and has been extended  to handle more than two data views \citep{witten2009penalized} with a number of variants being proposed (\cite{min2020sparse}, \cite{akaho2006kernel}, \cite{becker1992self}). CCA-based methods assume that all data views are generated from a common latent subspace and thus ignore individual variation specific to each data view. With more than two data views available, it might be also of interest to identify variation shared by only a subset of data views.
To remedy the limitations of CCA, 
the joint and individual variations explained (JIVE, \citep{lock2013joint}) method proposes a decompsoition model in which joint variation shared by all data views as well as individual variation specific to each data view are taken into account. 
Several variants of JIVE (COBE \citep{zhou2015group}; AJIVE \citep{feng2018angle}) have also been proposed to provide identifiable JIVE-type decomposition for multiview data. Going further, methods that can incorporate partially shared structures have been developed as well, including deterministic matrix decompositions (SLIDE, \citep{gaynanova2019structural}; \cite{jia2010factorized}) and probabilistic approaches  \citep{klamiGFA}. Specifically, deterministic models (SLIDE, \cite{gaynanova2019structural}; \cite{jia2010factorized}) allow different types of structures (joint/partially shared/individual) by inducing structured sparsity on the loading matrix while probabilistic models \citep{klamiGFA} achieve so by introducing a structural sparsity prior. Recently, two matrix decomposition based approaches, BIDIFAC (\cite{park2020integrative}) and BIDIFAC+ (\cite{lock2020bidimensional}), have been developed for bidimensionally linked data that are matched both at the sample level and the feature level. The uni-dimensional version of BIDIFAC and BIDIFAC+ can be applied to extract different types of variations for multiview data as well. It is worth mentioning that a related but distinct area is joint analysis of linked data,
where different populations of subjects are measured on the same set of features; see, e.g., \cite{flury1984common}, \cite{wang2021semiparametric} and \cite{chen2022two}.

In this work, our goal is to develop a theoretically sound approach for the joint analysis of multiview data. Motivated by the fact that the integrative approaches based on deterministic matrix decomposition (JIVE, SLIDE, BIDIFAC, BIDIFAC+) mentioned above can be essentially regarded as generalizations of traditional principal component analysis with different sparsity patterns (\cite{lock2013joint}, \cite{gaynanova2019structural}), we aim to develop a method directly using the framework of sparse PCA, for which there is abundant body of works that can be borrowed and statistical properties are much more tractable. 



As an important tool for dimension reduction, sparse PCA has attracted lots of attention in the last two decades and we refer to \cite{zou2018selective} for an extensive literature review on sparse PCA methods. One dominant line of work formulates sparse PCA as a an optimization problem with a non-smooth objective function and non-convex constraints. Some important approaches in this field include variance-maximization based formulation (\cite{jolliffe2003modified}, \cite{journee2010generalized}), regression based formulation \citep{zou2006sparse}, singular value decomposition based formulation (\cite{shen2008sparse}, \cite{witten2009penalized}). One computational challenge associated with the non-smooth and non-convex formulation is that some existing numerical algorithms are unable to provide convergence guarantee. Recently there have been a few works that introduce new algorithms that can solve non-convex problems with guarantee of convergence to a stationary point (\cite{chen2019alternating}, \cite{erichson2020sparse}). But a good starting point is still crucial for these algorithms to work out (\cite{chen2019alternating}). The other line of sparse PCA work is based on convex relaxation. \cite{d2004direct} formulates sparse PCA as a semidefinite programming problem using convex relaxation. 
\cite{vu2013fantope} develops Fantope projection and selection (FPS) with good computational and theoretical properties \citep{lei2015sparsistency}. FPS focuses on the sparsity of the whole principal subspace and does not allow individual eigenvectors to have different support regions. 
\cite{chen2015localized} addresses this issue by proposing a deflated Fantope based method, which sequentially estimates the eigenvectors and allows different support regions across eigenvectors. 

We propose a new sparse and integrative principal component analysis (SIPCA) for multiview data.
We concatenate multiview data vectors into a single multivariate data vector with variables within each data view grouped together as blocks and consider the popular spiked covariance model.
To accommodate hetereogenity in the data as the views might measure different characteristics, we allow the variance of errors to differ across views.
We formulate a convex optimization problem through Fantope (\citep{vu2013fantope}) and impose a combination of penalties to induce two-fold sparsity on the eigenvectors: elementwise sparsity and blockwise sparsity. As shall be shown later, an eigenvector with blocks of zeros corresponding to variation shared by a subset of views and thus such sparsity patterns are desired as they reveal the associations among the views, similar to the matrix decomposition methods such as JIVE (\citep{lock2013joint}) and SLIDE (\citep{gaynanova2019structural}). The Fantope based formulation is appealing for several reasons. First, it guarantees a convergent algorithm. Indeed, it allows us to develop an alternating direction method of multipliers (ADMM, \cite{boyd2011distributed}) algorithm for estimation, which guarantees convergence under mild conditions. Second, it could naturally accommodate data heterogeneity and allows us to use a denoised sample covariance in the formulation; see Section \ref{Methods} for details. 
Third, it enables us to investigate the theoretical properties of SIPCA. We have derived the $\ell_2$ rate of convergence of the estimated eigenvectors from SIPCA and also established their sparsity and support recovery properties. To our best knowledge, this is the first work with theoretical guarantee in the multiview data analysis literature.
Moreover, the theoretical work also contributes to the PCA literature when multiple penalties are presented.

As we are preparing this manuscript, we have become aware that \cite{zhang2021interpretable} proposes a similar Fantope based method for multivariate functional data. While our focus is on multiview data, the methods are similar. One major methodological difference is that we denoise the sample covariance to accommodate data heterogeneity which is necessary for multiview data. However, such a step might not be needed for functional data as smoothing can be used to remove errors. Another important difference is that \cite{zhang2021interpretable} does not have a theoretical study to validate their method, while  theoretical validation is a major focus of our work. It is  worth noting that
the theoretical properties established in our work can be applied to multivariate functional data, thus providing a theoretical support of \cite{zhang2021interpretable} for multivariate functional data.

The rest of this manuscript is structured as follows. In Section \ref{Methods}, we formulate the SIPCA and describe the estimation procedure and the optimization algorithm in detail. In Section \ref{theory}, we provide a theoretical analysis for the proposed estimate. Section \ref{simulation} presents extensive simulations comparing SIPCA with existing methods. In section \ref{application}, we apply SIPCA to the aforementioned TCGA data. We conclude with a discussion in Section \ref{discussion}. Technical details and additional materials are provided in the Appendix.

\subsection{Notation}
For a vector or matrix, we use $^{\t}$ to denote its transpose.
For any vector $\ba = (a_i)$, $\|\ba\|_\infty = \max_i |a_i|$,
$\|\ba\|_1 = \sum_i |a_i|$, and $\|\ba\|_2 = \sqrt{\sum_i a_i^2}$.
For any matrix $A = (A_{ij})$, 
$\|A\|_{\infty,\infty} = \max_{i, j} |A_{ij}|$,
$\|A\|_{1,1} = \sum_{i, j} |A_{ij}|$,
and $\|A\|_2 = \sqrt{\sum_{i, j} A_{ij}^2}$.
For a square matrix $A = (A_{ij})$, $\tr(A) = \sum_i A_{ii}$.
For two matrices $A = (A_{ij})$ and $B= (B_{ij})$ of the same dimensions,
 $\langle A, B\rangle = \tr(AB^{\t})$. For any set $J$, let $J^c$ be the complement of $J$. For two sets $J_1$ and $J_2$, $J_1\setminus J_2 = J_1\cap J_2^c$. For any real value $x$, let $x_{+}=\max(0,x)$.
 Let $\I(\cdot)$ be the indicator function. For any nonnegative number $a$, it is assumed that $a/0 = +\infty$.

%% file: sections/Methods.tex
\section{Sparse spiked covariance model for multiview data}\label{Methods}

Consider multiview data with $I$ views.
We concatenate the data vector from all data views for one sample and consider the popular spiked covariance model (\cite{johnstone2001distribution}, \cite{baik2005phase}, \cite{baik2006eigenvalues}, \cite{paul2007asymptotics}, \cite{donoho2018optimal})
for the concatenated sample:
\begin{equation}
\label{spiked:general}
\begin{split}
\by &= \bx + \be,\\
\bx &= (\bx_1^{\t},\ldots, \bx_I^{\t})^{\t}, \,\, \mathbb{E} \bx = \bmu,\,\, \text{Cov}(\bx) = \Sigma_1,\\
\be & = (\be_1^{\t},\ldots, \be_I^{\t})^{\t},\,\, \mathbb{E} \be = 0, \,\,\text{Cov}(\be) = \Sigma_e,\\
\bx, &\, \be_1,\ldots, \be_I \text{ are mutually independent}.
\end{split}
\end{equation}
In the model, $\bx\in \R^p$ is the signal vector, $\bx_i\in \R^{p_i}$ is the sub vector corresponding to the $i$th data view, and $p = \sum_i p_i$. We shall discuss the signal covariance matrix $\Sigma_1$ corresponding to the signal $\bx$ later.
The noise vector $\be$ is similarly partitioned into $I$ sub noise vectors. 
For the noise covariance $\Sigma_e$, we allow the noise variances to vary across different data views. Specifically,
\begin{equation}
\label{heter_noise}
    \Sigma_e = \text{blockdiag}(\sigma_1^2I_{p_1}, \cdots, \sigma_I^2I_{p_I}),
\end{equation}
where $\sigma_i^2, 1\leq i \leq I$, might be distinct. Such an assumption seems reasonable or even necessary to accommodate data heterogeneity as the different data views may be measured differently. Let $\Sigma$ be the covariance of $\by$, then $\Sigma = \Sigma_1 + \Sigma_e$.

We assume that the signal covariance matrix $\Sigma_1$ is of  rank $r$ and has the eigendecomposition $\Sigma_1 = \sum_{j=1}^r \gamma_j \bv_j \bv_j^{\t}$,
where $\lambda_1\geq \lambda_2\geq\cdots\geq\lambda_r> 0$ are  eigenvalues
and $\bv_j\in \R^p$ are eigenvectors satisfying $\la \bv_j, \bv_k\ra = \I(j=k)$. To ensure identifiability of the eigenvectors, we shall assume that the eigenvalues are distinct.
Similar to the blockwise partition of the signal vector $\bx$, the eigenvector $\bv_j$ can be partitioned as $\bv_j = ((\bv_{j}^1)^{\t},\ldots, (\bv_{j}^I)^{\t})^{\t}$ with
$\bv_{j}^i\in \R^{p_i}$ corresponding to the $i$th data view.
The vector $\bv_{j}^{i}$ quantifies the contribution of the $i$th data view to the $j$th principal component. Indeed, if $\|\bv_{j}^i\|_2 = 0$, then the $i$th data view does not contribute to the variation in the signal along the direction of the $j$th principal component. 
Likewise, $\Sigma_1$ can be partitioned into  blocks of submatrices $\left\{ (\Sigma_1)^{k\ell} \right\}_{1\leq k, \ell \leq I }$ with the $(k,\ell)th$ block $(\Sigma_1)^{k\ell}$ being defined as follows:
\begin{equation*}
    (\Sigma_1)^{k\ell} = \sum_{j=1}^r \lambda_j \bv_{j}^{k}(\bv_{j}^{\ell})^{\t} \in \R^{p_k\times p_\ell}.
\end{equation*}
The matrix $(\Sigma_1)^{k\ell}$ quantifies the covariance in the signal vector between the $k$th and $\ell$th data views. Notice that for an arbitrary $p\times p$ matrix, we can always partition it into blocks of submatrices.


In previous works (BIDIFAC+ \cite{lock2020bidimensional}, SLIDE \cite{gaynanova2019structural}), the variation of the multiview data is decomposed into three types: (i) joint variation shared by all data views; (ii) partially-shared variation shared by only a subset of data views; and (iii) individual variation unique to one data view. To that end, we impose blockwise sparsity structures on the eigenvectors. Take $I=3$ data views for instance. For an arbitrary eigenvector $\bv = \left((\bv^1)^{\t}, (\bv^2)^{\t}, (\bv^3)^{\t}\right)^{\t}$ which can be partitioned into three blocks, the potential blockwise sparsity structures it can possess include three types as listed  below:
\begin{itemize}
\item[(1)] Type I (joint structure): $\left((\bv^1)^{\t}, (\bv^2)^{\t}, (\bv^3)^{\t}\right)^{\t}$.
\item[(2)] Type II (partially shared structure): $\left(0^{\t}, (\bv^2)^{\t}, (\bv^3)^{\t}\right)^{\t}$, $\left((\bv^1)^{\t}, 0^{\t}, (\bv^3)^{\t}\right)^{\t}$, $\left((\bv^1)^{\t}, (\bv^2)^{\t}, 0^{\t}\right)^{\t}$.
\item[(3)] Type III (individual structure): $\left((\bv^1)^{\t}, 0^{\t}, 0^{\t}\right)^{\t}$, $\left(0^{\t}, (\bv^2)^{\t}, 0^{\t}\right)^{\t}$, $\left(0^{\t}, 0^{\t}, (\bv^3)^{\t}\right)^{\t}$.
\end{itemize}
The  variation along  an eigenvector with type I structure is contributed by all three data views and hence the variation is called joint variation; and type I structure of the eigenvector is referred to as a joint structure. 
As for type II structure, two views share a principal component score that is not present in the third view. Hence, an eigenvector of type II structure characterizes variation partially shared by two out of three views and
type II structure is referred to as a partially shared structure.
Lastly, eigenvectors with individual structures only capture variation specific to a single data view. Such a definition is straightforward to be generalized to any number of data views. Notice with $I\geq 4$ data views, a type II eigenvector may have $1$ to $I-1$  blocks of zeros.

Identification of the above three types of structures improves the interpretation of the principal component scores and help us understand
how the data views relate to each other. In particular, it enables us to quantify how much variation of a data view is shared with other data views and how much variation is exclusive to this data view.
In addition to the blockwise sparsity structure that are described above, we also would like to impose within-block variable sparsity. More precisely, we aim to identify relevant variables in each data view that contribute to the variation explained along a particular eigenvector.

We now introduce two types of sparsity index  to measure the elementwise and blockwise sparsity level of a vector of dimension $p$. First, for an arbitrary vector $\bv \in \R^p$, we define two index sets as below:
  \begin{equation*}
  \begin{split}
 &J_1(\bv)=\left\{i: v_i\neq 0\right\},\\
 &J_2(\bv)=\left\{k: \|\bv^{k} \|_2\neq 0\right\},
 \end{split}
 \end{equation*}
where $J_1(\bv)$ and $J_2(\bv)$ refer to the collection of indices corresponding to the elementwise and blockwise nonzero entries of $\bv$ respectively. Further, the cardinalities of $J_1(\bv)$ and $J_2(\bv)$ are denoted as $|J_1(\bv)|$ and $|J_2(\bv)|$ respectively.

Given $n$ independent and identically distributed copies of $\by$
from the spiked covariance model \eqref{spiked:general}, 
the aim is to estimate the eigenvectors of the signal covariance $\Sigma_1$
and also to identify their two levels of sparsity structures: blockwise sparsity and elementwise sparsity.


\subsection{Model estimation}
\label{sec:hetero}
The heteroskedasticity in noises across data views in \eqref{heter_noise} leads to one major challenge in the model estimation: the eigenvectors of covariance of the observed data, i.e., $\Sigma$ and those of the signal covariance $\Sigma_1$ are not necessarily the same. Thus, performing standard  PCA or sparsity-inducing PCA directly on the sample covariance matrix might give inconsistent estimates even when the dimension $p$ is small compared to the sample size $n$. As a remedy, \cite{florescu2016spectral} propose to replace the diagonal entries of the sample covariance matrix with zeros before performing PCA, which however again might change the targeted principal subspace, i.e., the space spanned by the eigenvectors of the signal covariances. 
\cite{zhang2022aos} develop HeteroPCA which estimates the signal principal subspace iteratively and in each iteration
the diagonal entries of the sample covariance matrix are replaced by estimates from the off-diagonals. However, HeteroPCA has two key drawbacks. First, it implicitly requires prior knowledge of the rank of the signal covariance matrix, which is often unknown. Second, HeteroPCA is an iterative algorithm without convergence guarantee. 

Our idea is to first denoise the sample covariance matrix, denoted as $\Sigmahat$. Note that the  covariance matrix of the $i$th  data view is a standard spiked covariance model with homoskedastic noise variance $\sigma_i^2$.
We adopt the method of bulk eigenvalue matching analysis (BEMA, \cite{ke2021estimation}) to each individual data view to obtain an estimate of $\sigma_i^2$, denoted as $\hat{\sigma}_i^2$. Briefly, for a spiked covariance model, BEMA utilizes a number of bulk empirical eigenvalues of the sample covariance in the middle range to estimate the theoretical distribution of sample eigenvalues of the signal covariance and then 
the noise estimate $\hat{\sigma}_i^2$ can  be found by minimizing the sum of squared difference between the bulk empirical eigenvalues and their theoretical limit. Compared to other related approaches (e.g., \cite{gavish2014optimal}, \cite{kritchman2009non}), BEMA is appealing because it seems to provide more accurate estimation empirically and its theoretical validity has also been established.

Given $\hat{\sigma}_i^2, 1\leq i \leq I$, from BEMA,  the additive structure in the spiked covariance model \eqref{spiked:general} motivates a straightforward pointwise estimate of $\Sigma_1$ shown as below,
\begin{equation*}
    S = \Sigmahat - \text{blockdiag}(\hat{\sigma}_1^2I_p, \cdots, \hat{\sigma}_I^2I_p).
\end{equation*}
The estimate $S$ can be viewed as a denoised version of the sample covariance $\Sigmahat$ and based on $S$, we consider estimating the eigenvectors $\bv_j (1\leq j \leq r$) associated with the signal covariance $\Sigma_1$. 

We start with estimation of the first  eigenvector $\bv_1$. Recall that we aim to induce two levels of sparsity in $\bv_1$, which can be accommodated by proposing a variance maximization based formulation with a combination of a lasso penalty and a group lasso penalty \citep{simon2013sparse}. Such a straightforward formulation might not be computationally appealing due to the non-convexity nature of the problem and it would become even more complex for higher-level eigenvector estimation problem which has an additional orthogonality constraint. So we follow the Fantope-based sparse PCA framework (\cite{vu2013fantope}, \cite{chen2015localized}, \cite{zhang2021interpretable}), and propose the following optimization problem:
\begin{equation*}
   \max_{H \in \F^{1}} \la S,H\ra-\lambda \beta \| H \|_{1,1} -\lambda (1-\beta) \|H\|_{1,1}^{\ast},
\end{equation*}
where $\F^1 = \{H\in \R^{p\times p}: 0\preceq H\preceq I_p \text{ and } \tr(H) = 1\}$ denotes the Fantope of degree one. Moreover, the first penalty  $\|H\|_{1,1}=\sum_{i=1}^{p}\sum_{j=1}^{p}|H_{ij}|$ is the matrix version of Lasso penalty, and the second penalty $\|H \|_{1,1}^{\ast}=\sum_{k=1}^{I}\sum_{\ell=1}^{I}w_{k\ell}\| H^{k\ell}\|_2$ is the matrix version of group lasso penalty.
Note that $H^{k\ell}$ denotes the $(k,\ell)$th sub block of $H$ and $w_{k\ell}$ is the weight that is used to adjust for the size of the submatrix $H^{k\ell}$.
The overall penalty is controlled by $\lambda\geq 0$ and 
$\beta$ is a number between 0 and 1 and balances the two types of penalties.
When it comes to estimation of higher level eigenvectors, i.e., $\bv_j \ (j\geq 2)$, we use the deflated Fantope constraint \citep{chen2015localized} to enforce strict orthogonality between different eigenvectors and the optimization problem is given as follows:
\begin{equation}
  \label{eq:model}
   \max_{H \in \F_{\hat{\Pi}_{j-1}}} \la S,H\ra-\lambda \beta \| H \|_{1,1} -\lambda (1-\beta) \|H\|_{1,1}^{\ast},
\end{equation}
where
$\F_{\hat{\Pi}_{j-1}} = \left\{H\in \R^{p\times p}:\, H\in \F^1 \text{ and } 
\la H, \hat{\Pi}_{j-1}\ra = 0\right\}$ is the deflated Fantope and $\hat{\Pi}_{j-1}= \sum_{k=1}^{j-1} \bvhat_k\bvhat_k^{\top}$ is the projection matrix associated with the principal subspace spanned by the first $j-1$ eigenvector estimates.
To summarize, the top $r$ eigenvector estimates $\bvhat_j, 1\leq j \leq r$ are obtained sequentially: 
\begin{equation*}
\begin{split}
\Hhat_j&=\operatorname{argmax}_{H\in\F_{\hat{\Pi}_{j-1}}}\left\{
\langle S,H \rangle-\lambda  \beta\|H \|_{1,1}-\lambda (1-\beta)\|H \|_{1,1}^{\ast}\right\},\\
\bvhat_j&=\text{the first eigenvector of } \Hhat_j,\\
\Pihat_j&=\Pihat_{j-1}+\bvhat_j\bvhat_j^{\t},
\end{split}
\end{equation*}
where $\Pihat_0= 0\in \R^{p\times p}$. Note that the tuning parameters $\lambda$ and $\beta$ don't have to be the same for all $j$. The optimization problem \eqref{eq:model} is convex and hence can be solved by existing algorithms with convergence guarantee. We shall describe an algorithm in detail in Section \ref{model:est}.  

This sequential formulation is critical because it leads to  mutually
orthogonal eigenvector estimates and also allows for different sparsity patterns in each individual eigenvector estimate. In contrast, the principal subspace approach in \cite{vu2013minimax} restricts all eigenvector estimates to have the same sparsity pattern and hence is not suitable here.

\subsection{Estimation algorithm} \label{model:est}

We use the alternating direction method of multipliers (ADMM, \cite{boyd2011distributed}) to solve the optimization problem \eqref{eq:model}. First, we rewrite \eqref{eq:model} by moving the deflated Fantope constraint to the objective function:
\begin{equation*}
    -\langle S,H \rangle+\lambda \beta \| H\|_{1,1}+\lambda (1-\beta)\| H \|_{1,1}^{\ast}+\infty \cdot (1-\I_{\mathcal{F}_{\Pi}}(H)),
\end{equation*}
where $\Pi$ is an arbitrary but fixed projection matrix, and $\I_{\mathcal{F}_{\Pi}}(H) = 1$ if $H\in \F_{\Pi}$ and 0 otherwise. The above problem can be reformulated using  variable splitting:
\begin{equation}
    \label{eq:model2}
\begin{split}
 &\min_{H_1,H_2}-\langle S,H_1\rangle+\infty \cdot (1-\I_{\mathcal{F}_{\Pi}}(H_1))+\lambda \beta\| H_2 \|_{1,1}+\lambda (1-\beta) \| H_2\|_{1,1}^{\ast},\\ &\text{ s.t. }  H_1-H_2=0.
\end{split}
\end{equation}
The augmented Lagrangian objective function for problem \eqref{eq:model2} is
\begin{equation*}
\begin{split}
Q(S;H_1,H_2,W)&=-\langle S,H_1\rangle+\infty \cdot (1-\I_{\mathcal{F}_{\Pi}}(H_1))+\lambda \beta \| H_2 \|_{1,1}+\lambda (1-\beta) \| H_2\|_{1,1}^{\ast}\\
& +\langle W,H_1-H_2\rangle+\frac{\rho_0}{2}\|H_1-H_2 \|_2^2,
\end{split}
\end{equation*}
where $W$ is matrix of dual variables and $\rho_0>0$ is called the penalty parameter. The ADMM algorithm alternates between the primal step, where $Q(S;H_1,H_2,W)$ is minimized with respect to $H_i$  $(i=1,2)$, and the dual step where the dual variable $W$ is updated. We denote the estimates from the previous iteration as $\tilde{H}_j$ $(j=1,2)$ and $\tilde{W}$. 

In the primal step, since the objective function $Q(S;H_1,H_2,W)$ is decomposable with respect to $H_1$ and $H_2$, they can be updated separately. The update for $H_1$ is given by
\begin{equation}\label{ag:step1}
\Hhat_1=\mathcal{P}_{\mathcal{F}_{\Pi}}\left\{\tilde{H}_2-\frac{1}{\rho_0}(\tilde{W}-S)\right\},
\end{equation}
where $\mathcal{P}_{\mathcal{F}_{\Pi}}$ is the deflated-Fantope-projection operator, of which the definition and the explicit expression are shown in Lemma \ref{df_ftp}. To update $H_2$, we need to solve the following subproblem 
$$
\min_{H_2}\frac{\rho_0}{2}\left\|\frac{1}{\rho_0}\tilde{W}+\hat{H}_1-H_2 \right\|_2^2 +\lambda \beta \|H_2 \|_{1,1}+\lambda (1-\beta)\|H_2 \|_{1,1}^{\ast},
$$
which also has a closed-form solution due to the fact that these two penalty terms possess a hierarchical structure. Specifically, denote the solution as $\hat{H}_2=\text{Prox}\left\{\frac{1}{\rho_0}\tilde{W}+\hat{H}_1\right\}$, and $\hat{H}_2$ can be obtained via two steps \cite{bach2012structured}:
\begin{equation}\label{ag:step2}
\begin{split}
    A&=\mathcal{S}_{\frac{\lambda \beta}{\rho_0}}\left\{\frac{1}{\rho_0}\tilde{W}+\hat{H}_1\right\},\\
   (\Hhat_2)^{k\ell}&=\left(1-\frac{\lambda (1-\beta)w_{k\ell}}{\rho_0\|A^{k\ell} \|_2} \right)_{+} A^{k\ell},
\end{split}
\end{equation}
where $\mathcal{S}_{\lambda}(x)=\operatorname{sgn}(x)(|x|-\lambda)_{+}$. In the dual step, the Lagrangian multiplier $W$ is updated as $\hat{W}=\tilde{W}+\rho_0(\Hhat_1-\Hhat_2)$. 

Although the global convergence of ADMM algorithm is guaranteed for convex problems, the performance of ADMM is often sensitive to the tuning parameter $\rho_0$ and
the convergence rate can be slow with a fixed $\rho_0$. A locally adaptive ADMM (LA-ADMM) proposed in
\cite{xu2017no} runs multiple stages of ADMM in which
at each stage the penalty parameter $\rho_0$ is dynamically increased by a constant factor $\alpha_0>1$ and the ADMM is run for a constant number of iterations.
It has been proved theoretically that LA-ADMM improves the computational speed of ADMM. Empirically, we did a small-scale simulation study, which shows that LA-ADMM is 2-8 times faster than ADMM. Hence, LA-ADMM is adopted to accelerate  convergence in our simulations and data analysis.

\medskip

\subsection{Tuning parameter selection}
In the optimization problem \eqref{eq:model}, the two tuning parameters $\lambda$ and $\beta$ are associated with the sparsity penalty and need to be selected for each $j$ separately. 

For estimation of the $j$th eigenvector, a five-fold cross-validation is adopted to determine the best tuning parameters $\lambda^*$ and $\beta^*$. The data is divided into five folds, denoted by $\mathcal{P}_1,\mathcal{P}_2,...,\mathcal{P}_5$. Let $\Hhat^{(-v)}(\lambda, \beta)$ be the solution to \eqref{eq:model} with tuning parameter $(\lambda, \beta)$ and the sample covariance is calculated on data other than $\mathcal{P}_v$. Let $S^{v}$ be the sample covariance estimated from $\mathcal{P}_v$. Then the optimal tuning parameters $(\lambda^*, \beta^*)$ are chosen by maximizing the cross-validated inner product shown as below:
$$
(\lambda^*, \beta^*)=\operatorname{argmax}_{\lambda\in\mathcal{A}_{1}, \beta\in \mathcal{A}_{2}}\sum_{v=1}^{5} \hat{e}^{(-v)}(\lambda, \beta)^{\t} S^v \hat{e}^{(-v)}(\lambda, \beta),
$$
where $\hat{e}^{(-v)}(\lambda, \beta)$ is the top eigenvector of $\Hhat^{(-v)}$ and $\mathcal{A}_{1}$, $\mathcal{A}_{2}$ are candidate sets for  $\lambda$ and $\beta$ respectively. We set $\mathcal{A}_{2}=\left\{0,0.25,0.5,0.75,1\right\}$, and a sequence between 0 and the 95\% quantile of absolute values of off-diagonal elements of $S_j=(I-\hat{\Pi}_{j-1})S(I-\hat{\Pi}_{j-1})$ is used as $\mathcal{A}_{1}$, following \cite{chen2015localized}.  The upper bound imposed on $\mathcal{A}_2$ is sensible as it can be shown a too large $\lambda$ will lead to trivial solutions.

%% file: sections/Theoretical_properties/General_theoretical_results.tex
\section{Theoretical properties}\label{theory}
We shall derive the $\ell_2$ rate of the estimated eigenvectors from SIPCA under appropriate conditions. In addition, we shall establish the sparsity and support recovery properties of the estimated eigenvectors.

\subsection{Deterministic theoretical results} \label{sec:general}
\def\B{\mathbb B}
\def\Zhat{\hat{Z}}
\def\Htilde{\tilde{H}}
\def\Ztilde{\tilde{Z}}
\def\Utilde{\tilde{U}}
\def\Uhat{\hat{U}}
\def\supp{\textnormal{supp}}
\def\diag{\textnormal{diag}}
\def\bu{\mathbf{u}}
\def\Sigmahat{\hat{\Sigma}}
\def\Sigmatilde{\tilde{\Sigma}}
\def\bVhat{\hat{V}}

In this subsection, we derive 
deterministic conditions that will be useful to  establish the sparsity of the estimates and derive the upper bounds of the $\ell_2$ error. The low rank property of the signal matrix, as shown in Assumption \ref{as:eta0}, is the key assumption imposed throughout the section.  All the propositions in this section imposes neither distributional assumptions on the data  nor  assumptions  on the eigenvalues or sparsity of eigenvectors of the signal matrix, thus the derived conditions and  bounds are in the forms of inequalities or functions involving the estimates.  In the next subsection, statistical assumptions are imposed, which lead to validation of those deterministic conditions.

We first introduce some assumptions and notations.

\begin{assumption}\label{as:eta0}
For the spiked covariance model in \eqref{spiked:general},
 the signal covariance matrix $\Sigma_1$
has the eigendecomposition $\Sigma_1 = \sum_{j=1}^r \gamma_j \bv_j\bv_j^{\intercal}$, where $r$ is a fixed integer,  $\gamma_1 > \gamma_2 >\cdots > \gamma_r> 0$
are eigenvalues  and $\bv_j$ is the associated eigenvector for $\gamma_j$.
\end{assumption}

For the optimization problem \eqref{eq:model},
let $w_0=\min_{k,\ell} w_{k\ell}$,  $w_1=\max_{k,\ell} w_{k\ell}$.
Define for $j=1,\ldots, r-1$, the projection matrices $\Pi_{j} = \sum_{k=1}^{j} \bv_k \bv_k^{\intercal}$
and $\Pihat_j = \sum_{k=1}^j \bvhat_k \bvhat_k^{\intercal}$, where $\bvhat_k$ is the estimated $k$th eigenvector from \eqref{eq:model}. Also let $\Pi_0 = \Pihat_0 = 0_{p \times p}$. Then for $j=1,\ldots, r-1$, let $\epsilon_j =\max(\|\Pi_{j} - \Pihat_{j}\|_2, \|\bvhat_j - \bv_j\|_2)$, which quantifies the estimation error from the first $j$ eigenvectors.
Let $E = S - \Sigma_1$ be the error matrix, where $S$ is the denoised sample covariance in Section \ref{sec:hetero}.
We introduce two useful matrix norms $\|\cdot\|_{\infty, \infty}$ and $\|\cdot\|_{\infty, \infty}^{\ast}$:
\begin{equation*}
\|E \|_{\infty, \infty} = \max_{s, t} |E_{st}|,\quad  \|E \|_{\infty, \infty}^{\ast} = \max_{k,\ell} \frac{\|(E)^{k\ell} \|_2}{w_{k\ell}}.
\end{equation*}
Note that $\|\cdot\|_{\infty, \infty}$ and $\|\cdot\|_{\infty, \infty}^{\ast}$ are conjugate norms for $\|\cdot\|_{1, 1}$ and $\|\cdot\|_{1, 1}^{\ast}$, respectively.

 \begin{prop}[Deterministic error bound] \label{prop:1}
Suppose that Assumption \ref{as:eta0} holds, and
\begin{equation}\label{eq:cond6}
    \min\left( \frac{\| E\|_{\infty,\infty}}{\beta},  \frac{\|E\|_{\infty,\infty}^{\ast}}{1-\beta} \right) \leq \lambda.
\end{equation}
Define, for $j=1,\ldots, r$,
$$
\tilde{a}_j = 2\left(\lambda \beta|J_1(\bv_j)| +\lambda (1-\beta ) w_1 |J_2(\bv_j)|
 +   \sum_{k=1}^{j-1}\gamma_k \epsilon_k\right),
$$
and 
\begin{equation*}
    \begin{split}
 \tilde{b}_j = 4 \sqrt{2} \left\{\lambda \beta \left(2|J_1(\bv_j)|+\sum_{k=1}^{j-1}|J_1(\bvhat_k)|\right) +\lambda (1-\beta) w_1 \left(2|J_2(\bv_j)|+\sum_{k=1}^{j-1}|J_2(\bvhat_k)|\right)
+ \frac{1}{2} \sum_{k=1}^{j-1}\gamma_k
\epsilon_k \right\}
\epsilon_{j-1}+4\gamma_j \epsilon_{j-1}^2.
\end{split}
\end{equation*}
Then, for $j=1,\ldots, r$,
\begin{equation}
    \|\Hhat_j-\bv_j\bv_j^{\intercal} \|_2\leq \frac{2\tilde{a}_j}{\gamma_j - \gamma_{j+1}}+\sqrt{\frac{2\tilde{b}_j}{\gamma_j - \gamma_{j+1}}},
    \label{eq:Herror}
\end{equation}
and
\begin{equation}\label{d_eb}
  \|\bvhat_j -\bv_j \|_2\leq 2\sqrt{2}\left(\frac{2\tilde{a}_j}{\gamma_j - \gamma_{j+1}}+\sqrt{\frac{2\tilde{b}_j}{\gamma_j - \gamma_{j+1}}}\right).
\end{equation}
\end{prop}

For the leading eigenvector, i.e., $j=1$, we can make $\epsilon_0 = 0$, thus $\tilde{b}_1 = 0$. Also $\tilde{a}_1 = 2(\lambda \beta |J_1(\bv_1)|+\lambda (1-\beta) w_1 |J_2(\bv_1)|)$. Thus,
inequality \eqref{eq:Herror} simplifies to
\begin{equation*}
\frac{\gamma_1 - \gamma_{2}}{2}\lVert \hat{H}_1 -\bv_1\bv_1^{\t}  \rVert_2\\
\leq 2\left\{\lambda \beta |J_1(\bv_1)|+\lambda (1-\beta) w_1 |J_2(\bv_1)|\right\}.
\end{equation*}
Notice that this error bound shows that the $\ell_2$ error bound of $\bvhat_1$ does not rely on the sparsity of $\bvhat_1$.
In particular, if we consider problem \eqref{eq:model} with $\beta=1$, we have
\begin{equation*}
\| \hat{H}_1 -\bv_1\bv_1^{\t}  \|_2
\leq \frac{4\lambda |J_1(\bv_1)|}{\gamma_1-\gamma_2},
\end{equation*}
which is the deterministic error bound derived in \cite{vu2013fantope}. Notice that condition~\eqref{eq:cond6} can still be used for $\beta = 0$ or 1 as we use the convention that for any nonnegative number $a$, $a/0 = +\infty$.

For the other eigenvectors, the forms of $\tilde{a}_j$
and $\tilde{b}_j$ in Proposition \ref{prop:1} involves $\epsilon_1,\ldots, \epsilon_{j-1}$, which quantify the estimation errors from the first $(j-1)$ eigenvectors. This is because of the deflation constraint in the optimization problem \eqref{eq:model}, which leads to the accumulation of estimation errors.

Next, we establish in Proposition \ref{prop:sparse} below that the estimated eigenvectors can be sparse,
a property that can be crucial for deriving the $\ell_2$ rate of convergence of the estimated eigenvectors.
Moreover, we provide deterministic conditions that guarantee the uniqueness of the solution to the optimization problem in \eqref{eq:model}, a non-trival result as the objective functions in \eqref{eq:model} is only convex but not strongly convex.



\begin{prop}[Deterministic sparsity control] \label{prop:sparse}
Suppose that Assumption \ref{as:eta0} holds.
For any fixed $r_1\leq r$, let $J_1 = \cup_{j=1}^{r_1} J_1(\bv_j)$ and $J_2 = \cup_{j=1}^{r_1} J_2(\bv_j)$.
Define
\begin{equation*}
b_j = \sum_{k=1}^{j} \frac{4[2(\lambda\beta  s_1+\lambda (1-\beta) s_2 w_1) + 2\gamma_k \epsilon_{k-1} + \gamma_1 \epsilon_{k-1}^2]}{\gamma_k-\gamma_{k+1}}, \text{ for } j = 1, \ldots, r_1,
\end{equation*}
where $s_1 = |J_1|$ and $s_2 = |J_2|$. Suppose that $\lambda$ and $0 < \beta \leq 1$ satisfy that,
for $1\leq j\leq r_1$,
\begin{equation}\label{eq:cond3}
\I(\beta\neq 1)\left\{
\|E\|_{\infty,\infty}^{\ast} + b_j\left\|(\Sigma_1)^{J_2^c J_2}\right\|_{2,\infty}^{\ast}\right\}
\leq \frac{w_0}{w_1}\lambda,
\end{equation}
\begin{equation}\label{eq:cond4}
\frac{1}{\beta}\left\{\| E \|_{\infty,\infty} + b_j\left\|(\Sigma_1)_{J_1^c J_1}\right\|_{2,\infty}\right\} \leq \lambda,
\end{equation}
\begin{equation}\label{eq:cond5}
0 <\gamma_j - \gamma_{j+1} - 
4\left(\lambda \beta s_1+\lambda (1-\beta) s_2 w_1 +\gamma_1 b_j \right),
\end{equation}
where 
\begin{equation*}
\begin{split}
    \|(\Sigma_1)^{J_2^cJ_2}\|_{2,\infty}^{\ast}&=\max_{\ell\in J_2^c} \sqrt{\sum_{k\in J_2}\frac{\|(\Sigma_1)^{k\ell} \|_2^2}{w_{k\ell}^2}}, \quad \|(\Sigma_1)_{J_1^cJ_1} \|_{2,\infty}=\max_{t\in J_1^c}\|(\Sigma_1)_{J_1t} \|_2.
\end{split}
\end{equation*}
Then  $\bvhat_j$ uniquely exists  and satisfies $J_1(\bvhat_j)\subseteq J_1$,\, $J_2(\bvhat_j)\subseteq J_2$,  for $j=1,\ldots, r_1$.
\end{prop}

\begin{remark}
The proof extends that in  \cite{lei2015sparsistency}. Notice that to prove $J_1(\bvhat_j) \subseteq J_1$ and $J_2(\bvhat_j) \subseteq J_2$, we may assume that $J_1(\bvhat_k) \subseteq J_1$ and $J_1(\bvhat_k) \subseteq J_2$ for $k=1,\ldots, (j-1)$.
\end{remark}

If $r_1 = r$, then $J_1$ contains all indices of non-zero elements in the eigenvectors $\{\bv_1,\ldots, \bv_r\}$ while $J_2$ contains all indices of non-zero blocks in these eigenvectors. It follows that
$(\Sigma_1)^{J_2^c J_2}$ consists of only zero elements, i.e., $\left\|(\Sigma_1)^{J_2^c J_2}\right\|_{2,\infty}^{\ast} =0$ and also 
$\left\|(\Sigma_1)_{J_1^c J_1}\right\|_{2,\infty}=0$, which simplify 
conditions \eqref{eq:cond3} and \eqref{eq:cond4} in Proposition \ref{prop:sparse}.

Next we consider a special case of problem \eqref{eq:model} with $\beta = 1$, i.e., only elementwise penalty is employed. 

\begin{corollary}\label{coro:elem}
Consider the optimization problem in \eqref{eq:model} with $\beta = 1$.
Suppose that Assumption \ref{as:eta0} holds.
For any fixed $r_1\leq r$, let $J_1 = \cup_{j=1}^{r_1} J_1(\bv_j)$.
 Suppose that  $\lambda$ satisfies, for $1\leq j\leq r_1$,
\begin{equation*}
\| E \|_{\infty,\infty} + b_j  \|(\Sigma_1)_{J_1^cJ_1} \|_{2,\infty}\leq \lambda,
\end{equation*}
\begin{equation*}
0 <\gamma_j - \gamma_{j+1} - 
4\left(\lambda s_1 + \gamma_1 b_j \right).
\end{equation*}
Then $\hat{\bv}_j$ uniquely exists and satisfies $J_1(\hat{\bv}_j)\subseteq J_1$, for $j=1,\ldots,r_1$.
\end{corollary}

If we consider only the leading eigenvector  $\bv_1$,
the conditions in Corollary \ref{coro:elem} will be  reduced to:
\begin{equation*}
\frac{1}{\lambda} \left(\|E \|_{\infty,\infty} + \frac{8s_1}{\gamma_1 - \gamma_2}  \|(\Sigma_1)_{J_1^cJ_1} \|_{2,\infty}\right)\leq 1,
\end{equation*}
\begin{equation*}
0 <\gamma_1 - \gamma_{2} - 
4\lambda s_1\left( 1+ \frac{8\gamma_1}{\gamma_1-\gamma_2} \right),
\end{equation*}
which are similar to the conditions in \cite{lei2015sparsistency} for proving sparsity of estimated principal spaces.

Proposition \ref{prop:sparse} does not cover the case of problem \eqref{eq:model} with $\beta=0$, but similar blockwise sparsity results can be obtained by a slight modification of our proof and the detailed proofs are omitted. 

\begin{corollary}\label{coro:block}
Consider the optimization problem in \eqref{eq:model} with $\beta = 0$.
Suppose that Assumption \ref{as:eta0} holds.
For any fixed $r_1\leq r$, let $J_2 = \cup_{j=1}^{r_1} J_2(\bv_j)$.
 Suppose that $\lambda$ satisfies, for $1\leq j\leq r_1$,
\begin{equation*}
\|E\|_{\infty,\infty}^{\ast}+ b_j \left\|(\Sigma_1)^{J_2^cJ_2}\right\|_{2,\infty}^{\ast}\leq \frac{w_0}{w_1}\lambda,
\end{equation*}
\begin{equation*}
0 <\gamma_j - \gamma_{j+1} - 
4\left(\lambda s_2 w_1 + \gamma_1 b_j\right).
\end{equation*}
Then $\hat{\bv}_j$ uniquely exists and satisfies $J_2(\hat{\bv}_j)\subseteq J_2$, for $j=1,\ldots,r_1$.
\end{corollary}







Finally, we discuss results on false negative control. Given the estimated projection matrix $\hat{H}_j$, we provide conditions so that there are no false negative element/blocks in the estimation of $\bv_j$.

\begin{prop}[Elementwise false negative control]\label{supp:element} 
 If
$
\min_{i:v_{ji}\neq 0}v_{ji}^2 > 2\|\Hhat_j-\bv_j\bv_j^{\t} \|_2,
$
then  $J_1(\bv_j)\subseteq J_1(\bvhat_j)$, i.e., there is no false negative element.
\end{prop}

\begin{prop}[Blockwise false negative control]\label{supp:block}
If $
\min_{k:\,\|\bv_j^{k}\|_2\neq 0}\|\bv_j^{k}\|_2^2 > 2\|\Hhat_j-\bv_j\bv_j^{\t} \|_2, 
$
then  $J_2(\bv_j)\subseteq J_2(\bvhat_j)$, i.e., there is no false negative block.
\end{prop}

%% file: sections/Theoretical_properties/Spiked_covariance_model_with_strong_signals.tex
\subsection{Statistical theoretical results}
\subsubsection{Assumptions and notations}
In this section, we make some additional assumptions in addition to Assumption \ref{as:eta0}  and introduce notation that will be used in the next two subsections.

For simplicity, we shall assume equal block size $p_0=p_1=\cdots=p_I$ and adopt equal block weight $w_{k\ell}=p_0, \forall k,\ell$. Thus, $w_0 = w_1 = p_0$.
We derive the statistical properties of the SIPCA estimator in a generic setting, where the elementwise error is bounded with a high probability. 
\begin{assumption}\label{as:prob_ieq}
Suppose that the random sample $\by = (y_1,\ldots, y_p)$ in model \eqref{spiked:general} is sub-Gaussian, i.e., there exists a constant $K>0$ such that
$$
\sup_{j} \mathbb{E}[\exp(y_j^2/K^2]\leq 2.
$$
\end{assumption}
Assume that $\sqrt{\log p/n} = o(1)$,
then under the sub-Guassian assumption, the sample covariance matrix $\Sigmahat$ satisfies 
\begin{equation*}
    \mathbb{P}\left(\|\Sigmahat - \mathbb{E}(\Sigmahat)\|_{\infty,\infty}\geq c \sqrt{\frac{\log p}{n}} \right) \geq 1- 2p^{-2},
\end{equation*}
for some absolute constant $c>0$; see, e.g., \cite{vaart1997weak}. 
Moreover,  the BEMA estimator $\hat{\sigma}_i^2$ satisfies that $\max_i \|\hat{\sigma}_i - \sigma_i^2\| = O_p(n^{-1})$; see Assumptions in Theorem 1 of \cite{ke2021estimation}.
It follows that the error matrix $E$ satisfies
\begin{equation*}
\|E\|_{\infty,\infty} = O_p\left(\sqrt{\frac{\log p}{n}} \right).
\end{equation*}
Note that as we assume equal block size and equal weight, $\|E\|_{\infty, \infty}^{\ast} \leq \|E\|_{\infty,\infty}$. Thus Assumption \ref{as:prob_ieq} implies that $\|E\|_{\infty,\infty}^{\ast} = O_p(\sqrt{\log p/n})$.

For notation, we let $J_1 = \cup_{j=1}^r J_1(\bv_j)$ and $J_2 = \cup_{j=1}^r J_2(\bv_j)$, $s_1 = |J_1|$ and $s_2 = |J_2|$. To compare, $J_1$ and $J_2$ as well as $s_1$ and $s_2$ vary with $r_1$ in Proposition \ref{prop:sparse}.
For two sequences of positive numbers $a_n$ and $b_n$,
$a_n \gtrapprox b_n $ means that there exists a sufficiently large constant $c$ such that $a_n \geq C_n b_n$ for all sufficiently large $n$ and yet still $a_n = O(b_n)$.

\subsubsection{Spiked covariance model with strong signals}\label{sec:strong}
We first shall consider the scenario that the signal  in the data is strong in the sense that the eigenvalues in the signal covariance matrix are proportional to the dimension of the signal covariance matrix; see Assumption \ref{as:eta}.
This scenario includes high-dimensional latent factor models \citep{fan2017high} and  functional data models \citep{bunea2015sample}.

\begin{assumption}\label{as:eta}
The eigenvalues $\gamma_j$s in  the signal covariance matrix $\Sigma_1$ satisfy $\gamma_j = \eta_j p$ for $1\leq j\leq r$ and $\eta_1 >\eta_2>\cdots >\eta_r >0$ are fixed constants.
\end{assumption}

Assumption \ref{as:eta} usually holds for high-dimensional latent factor data \citep{fan2017high} and  functional data  \citep{bunea2015sample}.
Under Assumption \ref{as:eta}, \cite{bunea2015sample} shows that the sample covariance matrix of sub-Gaussian random vectors with $\Sigma_1$ as the population covariance is consistent in both the scaled Frobenius norm and the scaled operator norm.

\begin{theorem}[Strong signal: statistical error bound]\label{thm:rate_block} 
Consider the optimization problem \eqref{eq:model} with $0\leq \beta \leq 1$. 
Suppose that Assumptions \ref{as:eta0} -- \ref{as:eta} hold and $\sqrt{\log p/n} = o(1)$. Suppose also that $\lambda \gtrapprox \sqrt{\log p/n}$. Then
\begin{equation*}
 \|\bvhat_j-\bv_j \|_2 = O_p\left(\sqrt{\frac{\log p}{n}}\right),\quad \text{for } j=1,\ldots, r.
\end{equation*}
\end{theorem}

Theorem \ref{thm:rate_block} establishes a near-parametric rate of convergence for the estimated eigenvectors. The theorem provides theoretical validation for the work in \cite{zhang2021interpretable}, in which multivariate functional data is considered and $\Sigma$ is the discretized covariance matrix for multivariate functional data. Indeed, for multivariate functional data, it can be easily verified that Assumption \ref{as:eta} holds under mild conditions; see, for example, \cite{bunea2015sample}.
If in addition $\beta = 1$,
then Theorem \ref{thm:rate_block} reduces to theoretical results established in \cite{chen2015localized} for multivariate functional data. When $\beta = 0$, then only blockwise penalty is employed and the above theoretical result is new.

Note that Assumption \ref{as:eta} might still be too restrictive and can be relaxed. If the eigenvalues of $\Sigma_1$ actually satisfy $\gamma_j = \eta_j \tilde{p}$ for $1\leq j\leq r$ and $\eta_1>\eta_2>\cdots> \eta_r >0$, then the bound in  Theorem \ref{thm:rate_block} 
become $O_p( p/\tilde{p}\sqrt{\log p/n})$ and thus consistency still holds as long as $p/\tilde{p} \sqrt{\log p/n} = o(1)$.

\begin{theorem}[Strong signal: sparsity and blockwise false negative control I]\label{thm:recovery_I_strong}
Consider the optimization problem \eqref{eq:model} with $\beta = 0$. 
Suppose the same assumptions and conditions as in Theorem \ref{thm:rate_block} hold.
\begin{enumerate}[label=(\roman*)]
\item
For any sufficiently large $n$, with high probability, $\bvhat_j$ uniquely exists and satisfies $$ J_2(\bvhat_j) \subseteq J_2,  \quad \text{ for }   j=1,\ldots, r.$$
\item
If
$$
\min_{k: \|\bv_j^k\|_2\neq 0} \|\bv_j^k\|_2^2 \gtrapprox  \sqrt{ \frac{\log p}{n}}, \quad \text{ for }  j=1,\ldots, r,
$$
then for any sufficiently large $n$, with high probability, 
$$J_2(\bv_j) \subseteq J_2(\bvhat_j) \quad \text{ for } j=1,\ldots, r.$$
\end{enumerate}
\end{theorem}



\begin{theorem}[Strong signal: sparsity and elementwise false negative control II]\label{thm:recovery_II_strong}
Consider the optimization problem \eqref{eq:model} with $0 < c < \beta \leq 1$, where $c$ is a fixed constant. 
Suppose the same assumptions and conditions as in Theorem \ref{thm:rate_block} hold.
\begin{enumerate}[label=(\roman*)]
\item
For any sufficiently large $n$, with high probability, $\bvhat_j$ uniquely exists and satisfies $$ J_1(\bvhat_j) \subseteq J_1,  \quad \text{ for }   j=1,\ldots, r.$$
\item
If 
$$
\min_{k:\, v_{jk} \neq 0} v_{jk}^2 \gtrapprox \sqrt{ \frac{ \log p}{n}}, \quad \text{ for }  j=1,\ldots, r,
$$
then for any sufficiently large $n$, with high probability, 
$$J_1(\bv_j) \subseteq J_1(\bvhat_j) \quad \text{ for } j=1,\ldots, r.$$

\end{enumerate}
\end{theorem}

Theorem \ref{thm:recovery_II_strong} includes the special case with $\beta = 1$, i.e., only elementwise penalty is imposed. The same observation can be made for Theorem \ref{thm:recovery_II_weak}.


%% file: sections/Theoretical_properties/Spiked_covariance_model_with_weak_signals.tex
\subsubsection{Spiked covariance model with weak signals}\label{sec:weak}
In this section, we consider the high dimensional spiked covariance model in which the eigenvalues are bounded (Assumption \ref{as:sigma1_weak}) and sparsity of eigenvectors is needed
to construct a consistent estimator.

\begin{assumption}\label{as:sigma1_weak}
The eigenvalues in the the signal covariance matrix $\Sigma_1$ are all bounded. 
\end{assumption}

\begin{theorem}[Weak signal: statistical error bound I]\label{thm:rate_block_weak} 
Consider the optimization problem with $\beta = 0$.
Suppose that Assumptions \ref{as:eta0},  \ref{as:prob_ieq},  and \ref{as:sigma1_weak} hold. Suppose also that $s_2p_0\sqrt{\log p/n} = o(1)$ and $\lambda \gtrapprox \sqrt{\log p/n}$. Then
\begin{equation*}
 \|\bvhat_j-\bv_j \|_2 = O_p\left(s_2p_0\sqrt{\frac{\log p}{n}}\right),  \quad \text{ for }   j=1,\ldots, r.
\end{equation*}
\end{theorem}

As $s_2$ counts the number of data views for which there is at least one signal eigenvector $\bv_j$ with corresponding non-zero blocks,
t assumption that $s_2p_0\sqrt{\log p/n} = o(1)$ implies that the number of non-zero blocks in the signal eigenvectors have to bounded, which leads to sparsity of the signal eigenvectors.

\begin{theorem}[Weak signal: sparsity and blockwise false negative control I]\label{thm:recovery_I_weak}
Consider the optimization problem \eqref{eq:model} with $\beta = 0$. 
Suppose the same assumptions and conditions as in Theorem \ref{thm:rate_block_weak} hold.
\begin{enumerate}[label=(\roman*)]
\item
For any sufficiently large $n$, with high probability, $\bvhat_j$ uniquely exists and satisfies $$ J_2(\bvhat_j) \subseteq J_2,  \quad \text{ for }   j=1,\ldots, r.$$
\item
If
$$
\min_{k: \|\bv_j^k\|_2\neq 0} \|\bv_j^k\|_2^2 \gtrapprox \sqrt{ \frac{\log p}{n}}, \quad \text{ for }  j=1,\ldots, r,
$$
then for any sufficiently large $n$, with high probability, 
$$J_2(\bv_j) \subseteq J_2(\bvhat_j) \quad \text{ for } j=1,\ldots, r.$$
\end{enumerate}
\end{theorem}

Part ($ii$) of the above theorem implies that in order for any non-zero block of the signal eigenvector to be detected, i.e., with elements estimated as non-zero, the minimal signal strength of the signal eigenvectors has to be at least of the order $\sqrt{\log p/n}$ blockwise.

\begin{theorem}[Weak signal: statistical error bound II]\label{thm:rate_joint_weak} 
Consider the optimization problem \eqref{eq:model} with $0 <c < \beta \leq 1$, where $c$ is  a fixed constant. Suppose that Assumptions \ref{as:eta0}, \ref{as:prob_ieq},  and \ref{as:sigma1_weak} hold.
Suppose also that $s_1\sqrt{\log p/n} = o(1)$, $s_2 p_0 =O(s_1)$, and $\lambda \gtrapprox \sqrt{\log p/n}$.
Then
\begin{equation*}
 \|\bvhat_j-\bv_j \|_2 = O_p\left(s_1\sqrt{\frac{\log p}{n}}\right),
 \quad \text{ for } j=1,\ldots, r.
\end{equation*}
\end{theorem}

In Theorem \ref{thm:rate_joint_weak}, the assumptions $s_1\sqrt{\log p/n} = o(1)$, $s_2 p_0 =O(s_1)$ imply that the signal eigenvectors are sparse elementwise.
Regarding the leading eigenvector estimate $\bvhat_1$, the rate presented in Theorem \ref{thm:rate_joint_weak} coincides with what has been established in \cite{vu2013minimax} for the sparse estimate for the projector of the leading one-dimensional principal subspace. However, this rate is suboptimal compared to the optimal minimax rate established for sparse eigenvector and sparse principal subspace estimation in the literature, which is $O_p(\sqrt{s_1 \log p/n})$ (\cite{cai2013sparse}; \cite{vu2013minimax}).
It has been shown that there is an underlying trade-off between the statistical and the computational efficiency for sparse PCA methods, and  no known computationally efficient methods are capable of attaining the minimax rate optimal estimators (\citep{vu2013fantope}, \citep{berthet2013computational}). Considering that SIPCA can be solved in polynomial time, the extra factor $\sqrt{s_1}$ might be unavoidable.

\begin{theorem}[Weak signal: sparsity and elementwise false negative control II]\label{thm:recovery_II_weak}
Consider the optimization problem \eqref{eq:model} with $0 <c < \beta \leq 1$, where $c$ is  a fixed constant. 
Suppose the same assumptions and conditions as in Theorem \ref{thm:rate_joint_weak} hold.
\begin{enumerate}[label=(\roman*)]
\item
Then for any sufficiently large $n$, with high probability, $\bvhat_j$ uniquely exists and satisfies $$ J_1(\bvhat_j) \subseteq J_1,  \quad \text{ for }   j=1,\ldots, r.$$
\item
If
$$
\min_{k:\, v_{jk} \neq 0} v_{jk}^2 \gtrapprox \sqrt{ \frac{ \log p}{n}}, \quad \text{ for }  j=1,\ldots, r,
$$
then for any sufficiently large $n$, with high probability, 
$$J_1(\bv_j) \subseteq J_1(\bvhat_j) \quad \text{ for } j=1,\ldots, r.$$
\end{enumerate}
\end{theorem}

%% file: sections/Simulation.tex
\section{Simulation}\label{simulation}
In this section, we investigate the performance of SIPCA through simulation studies. We also apply BIDIFAC+ \cite{lock2020bidimensional}, SLIDE \cite{gaynanova2019structural},
JIVE \cite{lock2013joint},
iPCA \cite{tang2021integrated},
PMD \cite{witten2009penalized}
and standard sample PCA to the simulated data for comparison.

\subsection{Simulation settings} \label{sec:sim_setup}
We generate zero-mean multivariate normal data using the spiked covariance model in \eqref{spiked:general} with the noise covariance in \eqref{heter_noise}.
We consider $I = 20$ blocks and each block contains $p_0 = 50$ elements so that the total dimension $p=1000$.
For the signal covariance, $\Sigma_1=\sum_{j=1}^{3}\lambda_j \bv_j \bv_j^{\intercal}$ and 
 the eigenvectors, i.e., $\bv_j, 1\leq j \leq 3$, have two levels of sparsity. At the block level,
\begin{equation*}
    J_2(\bv_1)=\left\{1, 2, 3\right\},  J_2(\bv_2)=\left\{4, 5, 6\right\},
    J_2(\bv_3)=\left\{7, 8\right\},
\end{equation*}
which means that the first three blocks of $\bv_1$ are nonzero and all others are zero blocks. Similarly, the blockwise supports for $\bv_2$ and $\bv_3$ are $\left\{4, 5, 6\right\}$ and $\left\{7, 8\right\}$ respectively. At the element level, each nonzero block in $\bv_{j}$ have $\gamma \times 100\%$ nonzero elements.
To introduce heteroskedasticity into the data,
we follow \cite{zhang2022aos} to determine $\sigma_i^2, 1 \leq i \leq I$ in the noise covariance.  We first generate $I$ random variables from the standard uniform distribution:
$ u_k\stackrel{\text{i.i.d}}{\sim} \text{Unif}[0,1], 1\leq k \leq I
$. Then we set $\sigma_i^2 =\sigma^2Iu_i^{\alpha}/(\sum_{k=1}^{I}u_k^{\alpha})$, where $\sigma^2=I^{-1}\sum_{i=1}^{I}\sigma_i^2$ and can be regarded as the average noise level. The parameter $\alpha$ delineates the degree of heteroskedasticity, with a larger $\alpha$ indicating more heteroskedastic noise. In particular, $\alpha=0$ corresponds to homoskedastic noise. In the simulation, we vary $\alpha$ over $\left\{0, 2, 5, 8, 10, 13, 15\right\}$. Moreover, we generate data with both weak signal and strong signal. For the former, $(\lambda_1, \lambda_2, \lambda_3) = (40, 20, 10)$ and for the latter $(\lambda_1, \lambda_2, \lambda_3) = (400, 200, 100)$. Finally, the sample size $n$ varies over $\left\{200, 400\right\}$, and we fix both $\sigma^2$ and $\gamma$ at 0.5.

\subsection{Evaluation criteria}
The $\ell_2$ error of the subspace estimators, i.e., $\|VV^{\t} - \hat{V}\hat{V}^{\t} \|_2/\| VV^{\t}\|_2$ is used to measure the estimation performance for different methods. For SIPCA, PMD and sample PCA, the top three eigenvector estimates are concatenated to form the signal principal subspace estimator $\hat{V}=[\bvhat_1, \bvhat_2, \bvhat_3]$. The algorithms for JIVE, SLIDE and BIDIFAC+ produce an estimation of the signal matrix, of which the top three right singular vectors are extracted to form $\hat{V}$. For iPCA, the loading subspace is estimated by performing SVD on the product of the common score matrix given by the iPCA algorithm and the concatenated observational data matrix.

We also evaluate the performance of SIPCA in terms of elementwise and blockwise support recovery.
For each eigenvector estimate, we calculate the corresponding sensitivity (true positive rate) and specificity (true negative rate) at both the element level and block level.




\subsection{Simulation results}
The simulation results on $\ell_2$ error are presented in Fig \ref{fig:sintheta}. JIVE estimators have a substantially larger error than other estimators across all scenarios. Further investigations show that the large estimation error given by JIVE for the top three principal subspace is mostly contributed by the eigen-space formed by the third eigenvector and JIVE seems to have trouble in identifying the third eigenvector, i.e., the eigenvector corresponding to the smallest eigenvalue of the signal covariance matrix. In the scenario where the signal is weak ($\lambda_1=40, \lambda_2=20, \lambda_3=10$) and the sample size is small ($n=200$), SIPCA, SLIDE, BIDIFAC+ and PMD perform comparably well when the noise is homoskedastic, i.e., $\alpha =0$, while iPCA and Sample PCA yield relatively worse estimation errors, which are both around 0.50.

As the noise becomes more heteroskedastic across different data views, i.e., as $\alpha$ increases, the median estimation errors given by SLIDE, BIDIFAC+, PMD, sample PCA and iPCA increase at a much faster pace than what is given by SIPCA. The advantage of SIPCA is even more pronounced when $n=400$, where SIPCA consistently yields the smallest median error and median abslolute deviation. In addition, increasing the heteroskedasticity in the simulated data does not seem to affect the performance of SIPCA. By contrast, there's a gradual and substantial decline in the estimation accuracy for other methods as $\alpha$ increases. 
In the scenario where the signal is strong ($\lambda_1=400, \lambda_2=200, \lambda_3=100$), SIPCA, SLIDE and BIDIFAC+ perform comparably well and have smallest $\ell_2$ distance among all methods. PMD gives slightly larger estimation error than the above three methods, followed by iPCA and sample PCA, which performed similarly and reasonably well regardless of the sample size and the level of heteroskedasticity when the signal is strong. 

\begin{figure}
    \centering
\subfloat[Weak signal \label{fig:sintheta_weak}]{   \includegraphics[width=15cm, height=9cm]{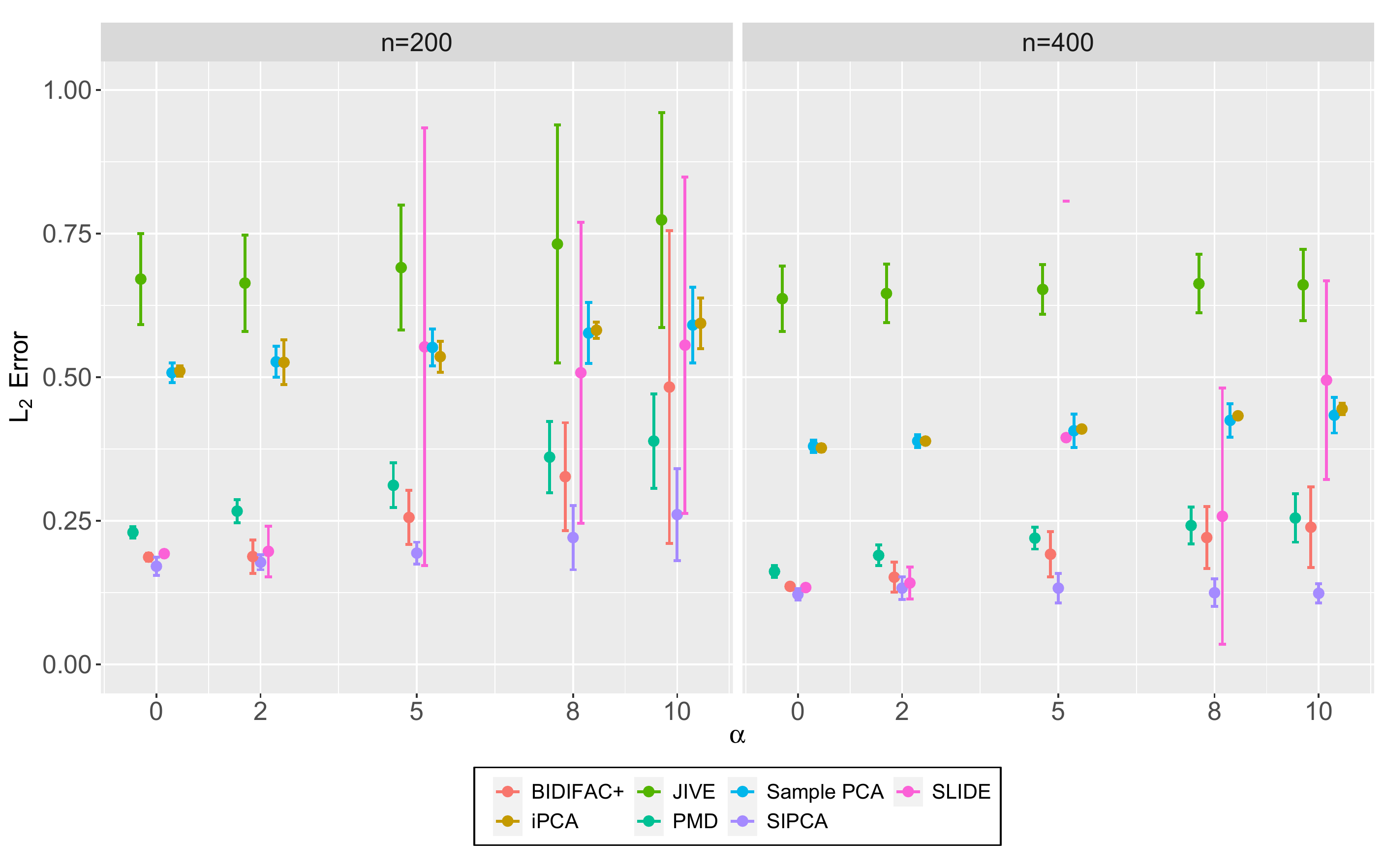}}

\subfloat[Strong signal \label{fig:sintheta_strong}]{    \includegraphics[width=15cm, height=9cm]{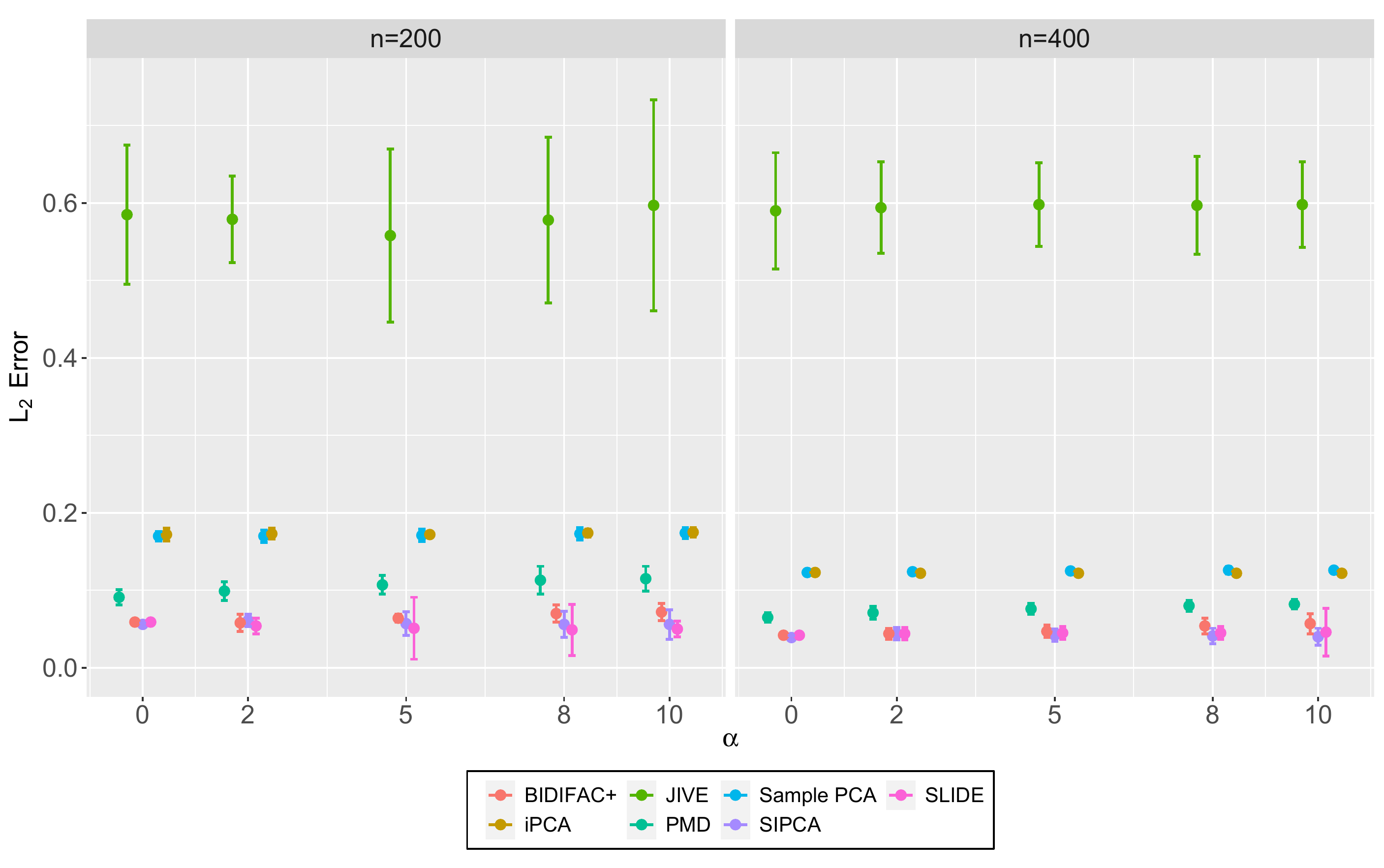}}
\caption{Median eigen subspace estimation error, i.e., $\|VV^{\t} - \hat{V}\hat{V}^{\t} \|_2/\| VV^{\t}\|_2$, over 50 replications for various methods on simulated data with different heteroskedastic noise level $\alpha$. The error bars indicate the corresponding median absolute deviation. Panel (a) and panel (b) present the results for the scenario with weak or strong signal respectively. Note that the x coordinate is jittered for better illustration. }
\label{fig:sintheta}
\end{figure}

The support recovery performance given by SIPCA for each individual eigenvector estimate in the setting with either weak or strong signal are presented In Table \ref{tb:sparsity_w} and Table \ref{tb:sparsity_s}. When the signal  is weak, increasing the level of heteroskedasticity in the noise tends to lead to lower true positive or true negative rate, especially for higher order eigenvectors. An increase in the sample size can significantly boost the performance. In the scenario with strong signal, the resulting sensitivity and specificity at both the block and the element level are greater than 0.9 across all simulation settings.

\renewcommand{\arraystretch}{2}
\begin{table}[H]
 \centering
  \fontsize{8}{6}\selectfont
  \begin{threeparttable}
  \caption{Median blockwise and elementwise sensitivity and specificity of eigenvector $\bvhat_j$ given by SIPCA in the simulation setting with weak signal. }
 \label{tb:sparsity_w}
    \begin{tabular}{ccccccccccccccccccccccc}
    \toprule
    \multicolumn{2}{c}{\multirow{3}{*}{$n=200$}} &
    \multicolumn{3}{c}{$\alpha=0$ }&\multicolumn{3}{c}{$\alpha=8$}&\multicolumn{3}{c}{$\alpha=15$}\cr
    \cmidrule(lr){3-5} \cmidrule(lr){6-8} \cmidrule(lr){9-11}
    & &$\bvhat_1$&$\bvhat_2$&$\bvhat_3$&$\bvhat_1$&$\bvhat_2$&$\bvhat_3$&$\bvhat_1$&$\bvhat_2$&$\bvhat_3$\cr
    \midrule
 \multirow{2}{*}{Block} &Sensitivity&1.00 &1.00 &1.00 &1.00 &1.00 &1.00 &1.00 &1.00 &1.00 \\
 &Specficity  &1.00 &1.00 &0.75 &1.00 &1.00 &0.83 &0.94 &0.82 &0.67\cr
\multirow{2}{*}{Element} &Sensitivity &0.95 &0.95 &0.99 &0.97 &0.97 &0.93 &0.99 &0.99 &0.95\\
&Specificity &0.95 &0.95 &0.79 &0.96 &0.95 &0.83 &0.92 &0.85 &0.79\cr
  \midrule
    \multicolumn{2}{c}{\multirow{3}{*}{$n=400$}} &
    \multicolumn{3}{c}{$\alpha=0$ }&\multicolumn{3}{c}{$\alpha=8$}&\multicolumn{3}{c}{$\alpha=15$}\cr
    \cmidrule(lr){3-5} \cmidrule(lr){6-8} \cmidrule(lr){9-11}

& &$\bvhat_1$&$\bvhat_2$&$\bvhat_3$&$\bvhat_1$&$\bvhat_2$&$\bvhat_3$&$\bvhat_1$&$\bvhat_2$&$\bvhat_3$\cr
\midrule
\multirow{2}{*}{Block}
&Sensitivity &1.00 &1.00 &1.00 &1.00 &1.00 &1.00 &1.00 &1.00 &1.00\\
&Specficity &1.00 &1.00 &0.94 &0.97 &1.00 &0.89 &0.94 &0.94 &0.83\cr
\multirow{2}{*}{Element} &Specificity &0.97 &0.97 &0.99 &0.99 &1.00 &0.98 &1.00 &1.00 &1.00\\
&Specificity&0.96 &0.95 &0.91 &0.96 &0.96 &0.87 &0.92 &0.92 &0.82\cr
    \bottomrule
    \end{tabular}
    \end{threeparttable}
\end{table}

\renewcommand{\arraystretch}{2}
\begin{table}[H]
 \centering
  \fontsize{8}{6}\selectfont
  \begin{threeparttable}
  \caption{ Median blockwise and elementwise sensitivity and specificity of eigenvector $\bvhat_j$ given by SIPCA in the simulation setting with strong signal.  }
 \label{tb:sparsity_s}
    \begin{tabular}{ccccccccccccccccccccccc}
    \toprule
    \multicolumn{2}{c}{\multirow{3}{*}{$n=200$}} &
    \multicolumn{3}{c}{$\alpha=0$ }&\multicolumn{3}{c}{$\alpha=8$}&\multicolumn{3}{c}{$\alpha=15$}\cr
    \cmidrule(lr){3-5} \cmidrule(lr){6-8} \cmidrule(lr){9-11}
    & &$\bvhat_1$&$\bvhat_2$&$\bvhat_3$&$\bvhat_1$&$\bvhat_2$&$\bvhat_3$&$\bvhat_1$&$\bvhat_2$&$\bvhat_3$\cr
    \midrule
    \multirow{2}{*}{Block} &Sensitivity &1.00 &1.00 &1.00 &1.00 &1.00 &1.00 &1.00 &1.00 &1.00\\
&Specficity &1.00 &1.00 &1.00 &1.00 &1.00 &1.00 &0.94 &1.00 &0.97\cr
    \multirow{2}{*}{Element} &Sensitivity 
 &1.00 &0.99 &0.98 &0.99 &1.00 &0.98 &0.99 &1.00 &0.98\\
 &Specificity &0.94 &0.96 &0.97 &0.96 &0.96 &0.97 &0.94 &0.96 &0.97\cr
  \midrule
\multicolumn{2}{c}{\multirow{3}{*}{$n=400$}} &
    \multicolumn{3}{c}{$\alpha=0$ }&\multicolumn{3}{c}{$\alpha=8$}&\multicolumn{3}{c}{$\alpha=15$}\cr
    \cmidrule(lr){3-5} \cmidrule(lr){6-8} \cmidrule(lr){9-11}
    & &$\bvhat_1$&$\bvhat_2$&$\bvhat_3$&$\bvhat_1$&$\bvhat_2$&$\bvhat_3$&$\bvhat_1$&$\bvhat_2$&$\bvhat_3$\cr
\midrule
    \multirow{2}{*}{Block} &Sensitivity &1.00 &1.00 &1.00 &1.00 &1.00 &1.00 &1.00 &1.00 &1.00\\
&Specficity &1.00 &1.00 &1.00 &1.00 &1.00 &1.00 &1.00 &1.00 &0.97\cr
    \multirow{2}{*}{Element} &Sensitivity  &1.00 &1.00 &0.99 &0.99 &0.99 &0.99 &1.00 &1.00 &1.00\\
  &Specificity &0.96 &0.94 &0.94 &0.98 &0.96 &0.96 &0.98 &0.97 &0.97\cr
    \bottomrule
    \end{tabular}
    \end{threeparttable}
\end{table}

%% file: sections/Real_data.tex
\section{Application to TCGA data}\label{application}
We apply SIPCA to multiview genomic data from The Cancer Genome Atlas (TCGA). In this dataset, four data views were generated on a common set of 348 tumor samples \citep{lock2013bayesian}:
\begin{itemize}
\item[(1)] RNA gene expression (GE) data for 645 genes,
\item[(2)] DNA methylation (ME) data for 574 probes,
\item[(3)] miRNA expression (miRNA) data for 423 miRNAs,
\item[(4)] Reverse phase protein array (RPPA) data for 171 proteins.
\end{itemize}
 These four data views were measured by different technologies on distinct platforms, but are still expected to be biologically correlated since they are drawn from the same set of samples. The raw data is publicly available on \url{https://www.cancer.gov/about-nci/organization/ccg/research/structural-genomics/tcga} and we use the preprocessed  dataset available in \url{https://github.com/ttriche/bayesCC/tree/master/data}. Among the 348 tumor samples, there are five breast cancer subtypes: 66 Basal-like samples, 42 
HER2-enriched samples, 154 Luminal A samples, 81 Luminal B samples and 5 normal-like samples. In this analysis, the normal-like samples are excluded due to the small size. 
We first scale each data matrix such that its Frobenius norm equals to the sample size (343) and then apply SIPCA, SLIDE, BIDIFAC+, iPCA, PMD and sample PCA to the concatenated data matrix. The rank estimate for the signal matrix given by SLIDE is 53, hence for fair comparison, the top 53 eigenvectors and the associated scores are extracted for SIPCA, iPCA sample PCA and PMD. For SIPCA, only blockwise penalty is implemented to induce sparsity structure at the view level, since the optimal $\beta$ by cross validation for most eigenvectors turned out to be $0$, which implies that the elementwise penalty might not add much additional value in this application. Similar to what was done in the simulation analysis for matrix decomposition based method, the scores used in downstream analysis are obtained by performing SVD on the signal matrix estimation for SLIDE, BIDIFAC+ and JIVE. 

We correlate the overall survival data for the 343 patients with tumors and the estimated scores through the cox proportional hazard model for each method. The table below displays the number of statistically significant scores identified in the modelling and we also report the concordance statistic which reflects the agreement between the observed survival time and the predicted survival time. Results show that most of the scores obtained by PMD or BIDIFAC+ do not have a strong statistical association with the survival outcome. On the contrary, JIVE, SIPCA, SLIDE, Sample PCA, and iPCA give a good number of scores ($\geq$ 8) that have statistically significant contribution to the survival prediction and the predicted survival data are highly aligned with the observed data (concordance $\geq$ 0.9).

\renewcommand{\arraystretch}{2}
\begin{table}[H]
\centering
\fontsize{10}{6}\selectfont
\caption{Table shows the number of statistically significant scores (confidence level being 0.01) and the concordance statistic for the model fitting with standard error in the parenthesis. } 
\label{tb:cluster_rand}
\begin{tabular}{ccccc}
   \toprule
 Method & Number of significant scores  & Concordance    \\
   \midrule
   SIPCA &9 & 0.914 (0.025) \\
   SLIDE &12 & 0.918 (0.019)   \\
   Sample PCA &14 &0.907 (0.029)  \\
   BIDIFAC+ &6 &0.861 (0.040)  \\
   JIVE&8 &0.906 (0.025)  \\
   PMD &2 &0.852 (0.048) \\
   iPCA &14 &0.916 (0.034) \\
   \bottomrule
\end{tabular}

\end{table}

%% file: sections/Discussion.tex
\section{Discussion}\label{discussion}
This work is motivated by the perceived need for a general and flexible approach with theoretical guarantee to modeling multiview data with potential heteroskedastic noise. We propose a novel sparse PCA method, SIPCA, which allows for flexible control on the sparsity structure of the resulting eigenvectors. Specifically, two levels of sparsity can be induced: blockwise sparsity reflecting the association structures between different data views and elementwise sparsity indicating important variables within each data view. The proposed method can be applied to sparse PCA problems where the noise is either homoskedastic or heteroskedastic. In the scenario where the noise levels for different data views are not identical, SIPCA employs a denoising step to obtain an estimate for the signal covariance, based on which the eigenvector estimates with desired sparsity are extracted. We demonstrate the superiority of SIPCA in a rich set of simulation settings including low to high degrees of noise heteroskedasticity and weak or strong signal strength. Compared to other competitors, SIPCA displays robust performance across all scenarios in terms of the estimation error for the signal principal subspace and the sparsity structure recovery for each individual eigenvector. Theoretical properties of the eigenvector estimate have been established from two perspectives: $\ell_2$ consistency and support recovery. Specifically, we show that the proposed eigenvector estimate achieves an $\ell_2$ rate comparable to what has been established in some related work (\cite{vu2013fantope}, \cite{vu2013minimax}).

We describe one major limitation of our method. As we mentioned before, it is natural for us to solve \eqref{eq:model} using ADMM which is a typical polynomial-time algorithm for constrained convex problems. However, the computational efficiency of ADMM has been shown to be low in our simulations despite the fact that it guarantees convergence to the global optima. Even if we implement an accelerated version of the vanilla ADMM \citep{xu2017no} in practice, there is still much room for improvement. This computational issue primarily relates to the full eigendecomposition of a $p\times p$ matrix per iteration, which requires high computational cost when the dimensionality is high. Recently, \cite{qiu2019gradient} propose a new gradient-based and projection-free algorithm with a focus on solving Fantope-based formulation for principal subspace estimation \citep{vu2013fantope}. Their empirical study indicates that the proposed algorithm takes much less time to achieve an estimate with certain level of accuracy compared with the original ADMM. While this algorithm cannot be directly applied to address our objective function \eqref{eq:model} since it is specifically designed for the optimization task proposed in \cite{vu2013fantope}, it is still of interest to adapt this new algorithm to accommodate different types of penalties.

%% file: sections/Appendix/Technical_proofs/Propositions/Proposition2.tex
\begin{proof}[Proof of Proposition \ref{prop:1}]
Let $\Sigma_j=\sum_{k=j}^{r}\gamma_k \bv_k\bv_k^{\intercal}$. Then, $\bv_j$ is the leading eigenvector of $\Sigma_j$.
By definition, $\hat{H}_j\in \mathcal{F}_{\hat{\Pi}_{j-1}}\subset\mathcal{F}^1$.
Let $\delta_j = \gamma_j - \gamma_{j+1}$. 
Then $\| \hat{H}_j-\bv_j\bv_j^{\intercal} \|_2$ can be bounded via Lemma \ref{lem:curv}:
\begin{align}
\frac{\delta_j}{2}\lVert\hat{H}_j -  \bv_j\bv_j^{\intercal}    \rVert_2^2\leq - \langle \Sigma_j, \hat{H}_j  -\bv_j\bv_j^{\intercal}
\rangle.\label{eq:deltaj}
\end{align}

Because $\bv_j\bv_j^{\t}\notin \F_{\hat{\Pi}_{j-1}}$,
$\bv_j\bv_j^{\t}$ is not a feasible solution for the optimization problem,  we define, for $j=1,\ldots, r$,
\begin{align*}
    \tilde{\bv}_j=\frac{(I-\hat{\Pi}_{j-1})\bv_j}{\lVert I-\hat{\Pi}_{j-1})\bv_j \rVert_2}.
\end{align*}
By Lemma \ref{lem:vtilde}, $\tilde{\bv}_j\tilde{\bv}_j^{\t}\in \F_{\hat{\Pi}_{j-1}}$ and is thus feasible. Note that
$\tilde{\bv}_j$ can be close to $\bv_j$ and its
properties  are established in Lemma \ref{lem:vtilde}.
We now rewrite \eqref{eq:deltaj} as 
\begin{align}
\frac{\delta_j}{2}\lVert   \hat{H}_j- \bv_j\bv_j^{\intercal} \rVert_2^2
\leq T_1+T_2,\label{eq:deltaj2}
\end{align}
where 
$T_1=\langle -\Sigma_j,\tilde{\bv}_j\tilde{\bv}_j^{\intercal}-\bv_j\bv_j^{\intercal}\rangle$ and
$T_2=\langle -\Sigma_j, \hat{H}_j- \tilde{\bv}_j\tilde{\bv}_j^{\intercal} \rangle$.

We first consider $T_1$.
Note that $\Sigma_j = \Sigma_{j+1} + \gamma_j \bv_j \bv_j^{\intercal}$ and $\langle \Sigma_{j+1}, \bv_j \bv_j^{\intercal}\rangle = 0$.
Thus,
\begin{align*}
T_1&=\langle -\Sigma_j,\tilde{\bv}_j\tilde{\bv}_j^{\intercal}-\bv_j\bv_j^{\intercal} \rangle \\
& =\langle -\Sigma_{j+1},\tilde{\bv}_j\tilde{\bv}_j^{\intercal}-\bv_j\bv_j^{\intercal} \rangle 
+ \langle -\gamma_j \bv_j \bv_j^{\intercal},\tilde{\bv}_j\tilde{\bv}_j^{\intercal}-\bv_j\bv_j^{\intercal} \rangle\\
&=\langle -\Sigma_{j+1},\tilde{\bv}_j\tilde{\bv}_j^{\intercal} \rangle 
+ \langle -\gamma_j \bv_j \bv_j^{\intercal},\tilde{\bv}_j\tilde{\bv}_j^{\intercal}-\bv_j\bv_j^{\intercal} \rangle\\
&\leq \langle -\gamma_j \bv_j \bv_j^{\intercal},\tilde{\bv}_j\tilde{\bv}_j^{\intercal}-\bv_j\bv_j^{\intercal} \rangle\\
&=\gamma_j \{1- (\bv_j^{\intercal}\tilde{\bv}_j)^2\}\\
&=\gamma_j\{1+(\bv_j^{\t}\tilde{\bv}_j)\}\{1-(\bv_j^{\t}\tilde{\bv}_j)\}\\
&\leq 2\gamma_j  \{1 - (\bv_j^{\intercal}\tilde{\bv}_j)\}\\
&\leq \gamma_j \|\bv_j - \tilde{\bv}_j\|_2^2.
\end{align*}
So we have derived that
\begin{equation}
    \label{eq:t1}
    T_1\leq \gamma_j \|\bv_j - \tilde{\bv}_j\|_2^2.
\end{equation}

We now consider $T_2$ and decompose $T_2$ into three terms:
\begin{align*}
   T_2&=\langle  S-\Sigma_1, \hat{H}_j-\tilde{\bv}_j\tilde{\bv}_j^{\intercal}    \rangle+\langle  \Sigma_1-\Sigma_j,\hat{H}_j-\tilde{\bv}_j\tilde{\bv}_j^{\intercal}      \rangle-\langle S,\hat{H}_j-\tilde{\bv}_j\tilde{\bv}_j^{\intercal}    \rangle.
\end{align*}
Let $T_3=\langle  S-\Sigma_1, \hat{H}_j-\tilde{\bv}_j\tilde{\bv}_j^{\intercal}    \rangle$, $T_4=\langle  \Sigma_1-\Sigma_j,\hat{H}_j-\tilde{\bv}_j\tilde{\bv}_j^{\intercal}      \rangle$ and  $T_5=\langle S,\hat{H}_j-\tilde{\bv}_j\tilde{\bv}_j^{\intercal}    \rangle $.
Then
\begin{equation}
    \label{eq:t2}
    T_2 = T_3 + T_4 - T_5,
\end{equation}
and we shall consider each of the above terms  separately.

First, we  analyze $T_4$. Note that
$\Sigma_1 - \Sigma_j = \sum_{k=1}^{j-1}\gamma_k \bv_k \bv_k^{\intercal}$.
By definition, for $1\leq k \leq j-1$,
$\langle \hat{\bv}_k\hat{\bv}_k^{\intercal}, \tilde{\bv}_j \tilde{\bv}_j^{\intercal}\rangle = 0$
and $\langle \hat{\bv}_k\hat{\bv}_k^{\intercal}, \hat{H}_j \rangle = 0$. Thus,
$\langle 
\Sigma_1 - \Sigma_j,
\hat{H}_j - \tilde{\bv}_j \tilde{\bv}_j^{\intercal}
\rangle = 0.
$
It follows that
\begin{align*}
T_4&=\langle \Sigma_1-\Sigma_j,\hat{H}_j-\tilde{\bv}_j\tilde{\bv}_j^{\intercal}  \rangle \nonumber\\
&=\langle  \sum_{k=1}^{j-1}\gamma_k\bv_k\bv_k^{\intercal}, \hat{H}_j-\tilde{\bv}_j\tilde{\bv}_j^{\intercal}  \rangle \nonumber\\
& = \langle  \sum_{k=1}^{j-1}\gamma_k (\bv_k \bv_k^{\intercal} - \hat{\bv}_k \hat{\bv}_k^{\intercal}), \hat{H}_j-\tilde{\bv}_j\tilde{\bv}_j^{\intercal}  \rangle \nonumber\\
&\leq 
 \left(\sum_{k=1}^{j-1}\gamma_k
\|\bv_k \bv_k^{\intercal} - \hat{\bv}_k \hat{\bv}_k^{\intercal}\|_2
\right)\left(
\|\hat{H}_j-\bv_j \bv_j^{\intercal}\|_2
+\|\bv_j \bv_j^{\intercal}-
\tilde{\bv}_j\tilde{\bv}_j^{\intercal} \|_2
\right).
\end{align*}
By Lemma \ref{lem:Hv}, 
$\|\bv_k\bv_k^{\intercal} - \hat{\bv}_k \hat{\bv}_k^{\intercal}\|_2\leq \sqrt{2} \|\bv_k - \hat{\bv}_k\|_2$
and $\|\bv_j \bv_j^{\intercal}-
\tilde{\bv}_j\tilde{\bv}_j^{\intercal} \|_2\leq \sqrt{2}\|\bv_j - \tilde{\bv}_j\|_2$.
Hence,
\begin{equation}
    \label{eq:t4}
    T_4\leq 
   \sqrt{2}  \left(\sum_{k=1}^{j-1}\gamma_k
\|\bv_k - \hat{\bv}_k \|_2
\right)\left(
\|\hat{H}_j-\bv_j \bv_j^{\intercal}\|_2
+\sqrt{2}\|\bv_j -
\tilde{\bv}_j \|_2
\right).
\end{equation}

Next, we analyze $T_5$. Because $\hat{H}_j$ is the optimal solution and $\tilde{\bv}_j\tilde{\bv}_j^{\intercal}$ is feasible, 
\begin{align*}
&\langle S,\hat{H}_j  \rangle-\lambda \beta \lVert \hat{H}_j\rVert_{1,1}-\lambda (1-\beta) \|\hat{H}_j  \|_{1,1}^{\ast}\geq \langle S,\tilde{\bv}_j\tilde{\bv}_j^{\intercal} \rangle-\lambda \beta \| \tilde{\bv}_j\tilde{\bv}_j^{\intercal}\|_{1,1}-\lambda (1-\beta)\| \tilde{\bv}_j\tilde{\bv}_j^{\intercal}\|_{1,1}^{\ast}.
\end{align*}
Hence,
\begin{align}
  -T_5=-\langle S,\hat{H}_j-\tilde{\bv}_j\tilde{\bv}_j^{\intercal}    \rangle \leq  \lambda \beta ( -\lVert \hat{H}_j \rVert_{1,1}+\lVert \tilde{\bv}_j \tilde{\bv}_j^{\t}\rVert_{1,1})+
    \lambda (1-\beta)(-\|\hat{H}_j  \|_{1,1}^{\ast}+\|\tilde{\bv}_j\tilde{\bv}_j^{\t} \|_{1,1}^{\ast}).
    \label{eq:t5}
\end{align}

We now consider $T_3$. 
First,
\begin{align*}
T_3&=\langle  S-\Sigma_1, \hat{H}_j-\tilde{\bv}_j\tilde{\bv}_j^{\intercal}    \rangle
\leq  \lVert S-\Sigma_1  \rVert_{\infty,\infty} \lVert \hat{H}_j-\tilde{\bv}_j\tilde{\bv}_j^{\intercal}   \rVert_{1,1}.
\end{align*}
Also,
\begin{align*}
T_3&=\langle  S-\Sigma_1, \hat{H}_j-\tilde{\bv}_j\tilde{\bv}_j^{\intercal}    \rangle\\
&=\sum_{k=1}^{I}\sum_{\ell=1}^{I}\langle \left(S-\Sigma_1\right)^{k\ell},(\hat{H}_j-\tilde{\bv}_j\tilde{\bv}_j^{\t})^{k\ell} \rangle\\
&\leq \sum_{k=1}^{I}\sum_{\ell=1}^{I} \frac{\|\left(S-\Sigma_1\right)^{k\ell} \|_2}{w_{k\ell}}\left(w_{k\ell}\left\| \left(\hat{H}_j-\tilde{\bv}_j\tilde{\bv}_j^{\intercal}\right)^{k\ell} \right\|_2\right)\\
&\leq \left(\max_{k,\ell}\frac{\|\left(S-\Sigma_1\right)^{k\ell} \|_2}{w_{k\ell}}\right)\sum_{k=1}^{I}\sum_{\ell=1}^{I}\left(w_{k\ell}\left\| \left(\hat{H}_j-\tilde{\bv}_j\tilde{\bv}_j^{\intercal}\right)^{k\ell} \right\|_2\right)\\
&\leq \|S-\Sigma_1\|_{\infty,\infty}^{\ast} \|\hat{H}_j-\tilde{\bv}_j\tilde{\bv}_j^{\t} \|_{1,1}^{\ast}.
\end{align*}
Thus, if condition \eqref{eq:cond6}  holds, 
\begin{align}
    T_3\leq \lambda \beta \lVert \hat{H}_j-\tilde{\bv}_j\tilde{\bv}_j^{\intercal}  \rVert_{1,1}+ \lambda (1-\beta) \|\hat{H}_j-\tilde{\bv}_j\tilde{\bv}_j^{\t} \|_{1,1}^{\ast}. \label{eq:t3}
\end{align}

Combining \eqref{eq:t5} and \eqref{eq:t3}, and by the triangle inequality, we have
\begin{align*}
    T_3-T_5&\leq \lambda \beta \left(\lVert \hat{H}_j-\tilde{\bv}_j\tilde{\bv}_j^{\t}\rVert_{1,1}+\lVert \tilde{\bv}_j \tilde{\bv}_j^{\t}\rVert_{1,1}-\lVert \hat{H}_j \rVert_{1,1} \right)+\lambda (1-\beta)(\|\hat{H}_j-\tilde{\bv}_j\tilde{\bv}_j^{\t} \|_{1,1}^{\ast}+\|\tilde{\bv}_j\tilde{\bv}_j^{\t}\|_{1,1}^{\ast}-\|\hat{H}_j  \|_{1,1}^{\ast})\nonumber\\
    &\leq \lambda \beta \left(\lVert \hat{H}_j-\bv_j\bv_j^{\intercal}\rVert_{1,1}+\lVert \bv_j \bv_j^{\t}\rVert_{1,1}-\lVert \hat{H}_j \rVert_{1,1} \right)+\lambda (1-\beta)(\|\hat{H}_j-\bv_j\bv_j^{\t} \|_{1,1}^{\ast}+\|\bv_j\bv_j^{\t}\|_{1,1}^{\ast}-\|\hat{H}_j  \|_{1,1}^{\ast})\nonumber\\
    &+2\lambda \beta \lVert \tilde{\bv}_j\tilde{\bv}_j^{\intercal}-\bv_j\bv_j^{\intercal} \rVert_{1,1}+2\lambda (1-\beta) \| \tilde{\bv}_j\tilde{\bv}_j^{\intercal}-\bv_j\bv_j^{\intercal} \|_{1,1}^{\ast}.
\end{align*}
Let 
\begin{align*}
   &T_6= \lVert \hat{H}_j-\bv_j\bv_j^{\intercal}\rVert_{1,1}+\lVert \bv_j \bv_j^{\intercal}\rVert_{1,1}-\lVert \hat{H}_j \rVert_{1,1},\\
   &T_7=\|\hat{H}_j-\bv_j\bv_j^{\t} \|_{1,1}^{\ast}+\|\bv_j\bv_j^{\intercal}\|_{1,1}^{\ast}-\|\hat{H}_j  \|_{1,1}^{\ast}.
\end{align*}
Then,
\begin{equation}
\label{eq:t3_5}
        T_3 - T_5 \leq \lambda \beta T_6+\lambda (1-\beta) T_7 +2\lambda \beta \lVert \tilde{\bv}_j\tilde{\bv}_j^{\intercal}-\bv_j\bv_j^{\intercal} \rVert_{1,1}+2\lambda (1-\beta)\| \tilde{\bv}_j\tilde{\bv}_j^{\intercal}-\bv_j\bv_j^{\intercal} \|_{1,1}^{\ast}.
\end{equation}

By Lemma \ref{lem:ineq}, we can bound $T_6$ and $T_7$ as follows:
\begin{align}
 &T_6\leq 2\sqrt{|J_1(\bv_j\bv_j^{\intercal})|}\|\hat{H}_j-\bv_j\bv_j^{\intercal} \|_2,\label{eq:t6}\\
 &T_7\leq 2\sqrt{\sum_{(k,\ell)\in J_2(\bv_j\bv_j^{\intercal})}w_{k\ell}^2}\|\hat{H}_j-\bv_j\bv_j^{\intercal} \|_2.\label{eq:t7}
\end{align}

Next, we use the Cauchy-Schwartz inequality to bound $\| \tilde{\bv}_j\tilde{\bv}_j^{\intercal}-\bv_j\bv_j^{\intercal} \|_{1,1}$ and $\|\tilde{\bv}_j\tilde{\bv}_j^{\intercal}-\bv_j\bv_j^{\intercal} \|_{1,1}^{\ast}$.
We derive that
\begin{align*}
\lVert \tilde{\bv}_j\tilde{\bv}_j^{\intercal}-\bv_j\bv_j^{\intercal} \rVert_{1,1}&\leq \sqrt{|J_1(\tilde{\bv}_j\tilde{\bv}_j^{\intercal}-\bv_j\bv_j^{\intercal})|}\lVert \tilde{\bv}_j\tilde{\bv}_j^{\intercal}-\bv_j\bv_j^{\intercal}\rVert_2.
\end{align*}
By Lemma \ref{lem:Hv}, 
$\lVert \tilde{\bv}_j\tilde{\bv}_j^{\intercal}-\bv_j\bv_j^{\intercal}\rVert_2\leq \sqrt{2}\lVert \tilde{\bv}_j-\bv_j\rVert_2$. Thus,
\begin{equation}
    \label{eq:t3_5_1}
    \lVert \tilde{\bv}_j\tilde{\bv}_j^{\intercal}-\bv_j\bv_j^{\intercal} \rVert_{1,1}
    \leq \sqrt{2} \sqrt{|J_1(\tilde{\bv}_j\tilde{\bv}_j^{\intercal}-\bv_j\bv_j^{\intercal})|}\lVert \tilde{\bv}_j-\bv_j\rVert_2.
\end{equation}
Similarly, we derive that
\begin{align}
\|\tilde{\bv}_j\tilde{\bv}_j^{\intercal}-\bv_j\bv_j^{\intercal} \|_{1,1}^{\ast}
\leq \sqrt{2}\sqrt{\sum_{(k,l)\in J_2(\tilde{\bv}_j\tilde{\bv}_j^{\intercal}-\bv_j\bv_j^{\intercal})}w_{kl}^2}\lVert \tilde{\bv}_j-\bv_j\rVert_2.
\label{eq:t3_5_2}
\end{align}
Substituting \eqref{eq:t6}, \eqref{eq:t7}, \eqref{eq:t3_5_1} and \eqref{eq:t3_5_2} into \eqref{eq:t3_5},
\begin{align*}
T_3-T_5&\leq 2\left(\lambda \beta \sqrt{|J_1(\bv_j\bv_j^{\intercal})|}+\lambda (1-\beta)\sqrt{\sum_{(k,\ell)\in J_2(\bv_j\bv_j^{\intercal})}w_{k\ell}^2}\right)\|\hat{H}_j-\bv_j\bv_j^{\intercal} \|_2\nonumber\\
&+2\sqrt{2} \left(\lambda \beta \sqrt{|J_1(\tilde{\bv}_j\tilde{\bv}_j^{\intercal}-\bv_j\bv_j^{\intercal})|}+\lambda (1-\beta) \sqrt{\sum_{(k,\ell)\in J_2(\tilde{\bv}_j\tilde{\bv}_j^{\intercal}-\bv_j\bv_j^{\intercal})}w_{k\ell}^2}\right)
\lVert \tilde{\bv}_j-\bv_j\rVert_2.
\end{align*}

Combining \eqref{eq:deltaj2},
\eqref{eq:t1}, \eqref{eq:t2}, \eqref{eq:t4} and
the above inequality, 
 \begin{equation}\label{eq:Herror0}
     \begin{split}
 &\frac{\delta_j}{2}\lVert  \hat{H}_j-\bv_j\bv_j^{\intercal}  \rVert_2^2\\
 &\leq \left(2\lambda\beta |J_1(\bv_j)|+2\lambda (1-\beta)\sqrt{\sum_{(k,\ell)\in J_2(\bv_j)\times J_2(\bv_j)}w_{k\ell}^2}
 + \sqrt{2}  \sum_{k=1}^{j-1}\lambda_k
\|\bv_k - \hat{\bv}_k \|_2\right)\|\hat{H}_j-\bv_j\bv_j^{\intercal} \|_2\\
&+2\sqrt{2} \left(\lambda \beta  \sqrt{|J_1(\tilde{\bv}_j\tilde{\bv}_j^{\intercal}-\bv_j\bv_j^{\intercal})|}+\lambda (1-\beta) \sqrt{\sum_{(k,\ell)\in J_2(\tilde{\bv}_j\tilde{\bv}_j^{\intercal}-\bv_j\bv_j^{\intercal})}w_{k\ell}^2}
+ \frac{1}{2}\sum_{k=1}^{j-1}\lambda_k
\|\hat{\bv}_k -\bv_k   \|_2 \right)
\lVert \tilde{\bv}_j-\bv_j\rVert_2\\
&+\lambda_j \lVert \tilde{\bv}_j - \bv_j \rVert_2^2.
\end{split}
 \end{equation}
 
We now further derive that
\begin{equation*}
\begin{split}
\sqrt{|J_1(\tilde{\bv}_j\tilde{\bv}_j^{\intercal}-\bv_j\bv_j^{\intercal})|}&\leq \sqrt{|J_1(\tilde{\bv}_j\tilde{\bv}_j^{\intercal})|+|J_1(\bv_j\bv_j^{\intercal})|}\\
&\leq |J_1(\tilde{\bv}_j)|+|J_1(\bv_j)|\\
&\leq 2|J_1(\bv_j)|+\sum_{k=1}^{j-1}|J_1(\bvhat_k)|,
\end{split}
\end{equation*}
and the second to last inequality is due to (c) in Lemma \ref{lem:vtilde}.

Similarly, we have
\begin{equation*}
\begin{split}
\sqrt{|J_2(\tilde{\bv}_j\tilde{\bv}_j^{\intercal}-\bv_j\bv_j^{\intercal})|}&\leq \sqrt{|J_2(\tilde{\bv}_j\tilde{\bv}_j^{\intercal})|+|J_2(\bv_j\bv_j^{\intercal})|}\\
&\leq |J_2(\tilde{\bv}_j)|+|J_2(\bv_j)|\\
&\leq 2|J_2(\bv_j)|+\sum_{k=1}^{j-1}|J_2(\bvhat_k)|.
\end{split}
\end{equation*}
Combining these two inequalities and using Lemma \ref{lem:vtilde}(a), we obtain
from \eqref{eq:Herror0} that
\begin{equation*}
     \begin{split}
 &\frac{\delta_j}{2}\| \hat{H}_j -\bv_j\bv_j^{\t}  \|_2^2\\
 \leq &2\left(\lambda \beta |J_1(\bv_j)|+\lambda (1-\beta) w_1 J_2(\bv_j)
 +   \sum_{k=1}^{j-1}\lambda_k
\epsilon_k \right)\|\hat{H}_j-\bv_j\bv_j^{\intercal} \|_2\\
& +4 \sqrt{2} \left\{\lambda \beta  \left(2|J_1(\bv_j)|+\sum_{k=1}^{j-1}|J_1(\bvhat_k)|\right) +\lambda (1-\beta) w_1 \left(2|J_2(\bv_j)|+\sum_{k=1}^{j-1}|J_2(\bvhat_k)|\right)
+ \frac{1}{2}\sum_{k=1}^{j-1}\lambda_k
\epsilon_k \right\}
\epsilon_{j-1}+4\lambda_j \epsilon_{j-1}^2.
\end{split}
 \end{equation*}
By the specifications of $\tilde{a}_j$ and $\tilde{b}_j$, the above inequality reduces to 
$$
\frac{\delta_j}{2}\| \hat{H}_j -\bv_j\bv_j^{\t}  \|_2^2
\leq \tilde{a}_j \| \hat{H}_j -\bv_j\bv_j^{\t}  \|_2 + \tilde{b}_j,
$$
which gives 
$$
\| \hat{H}_j -\bv_j\bv_j^{\t}  \|_2
\leq \frac{2 \tilde{a}_j}{\delta_j} + \sqrt{\frac{2\tilde{b}_j}{\delta_j}}.
$$
Then by Lemma \ref{lem:vhat} and Lemma \ref{lem:Hv}, we have
\begin{equation*}
\begin{split}
\|\bvhat_j-\bvhat_j\|_2 &\leq \sqrt{2}\|\bvhat_j\bvhat_j^{\t}-\bv_j\bv_j^{\t} \|_2\\
&\leq 2\sqrt{2}\|\Hhat_j-\bv_j\bv_j^{\t} \|_2\\
&\leq 2\sqrt{2}\left(\frac{2\tilde{a}_j}{\delta_j}+\sqrt{\frac{2\tilde{b}_j}{\delta_j}}\right).
\end{split}
\end{equation*}
Now the proof is complete.
\end{proof}

%% file: sections/Appendix/Technical_proofs/Propositions/Proposition1.tex
\textit{Proof of Proposition} \ref{prop:sparse}:
\begin{proof}[Existence of a sparse solution]  Let $$\B_{p,1}=\{Z\in \R^{p\times p}:Z=Z^{\t},\|Z \|_{\infty,\infty}\leq 1 \},$$
and 
$$\B_{p,2} = \{Z\in \R^{p\times p}:
Z = Z^{\t}, \|Z^{k\ell}\|_2 \leq w_{k\ell},\forall (k,\ell)
\}.$$
We write \eqref{eq:model} into an equivalent min-max form:
\begin{align*}
       &\quad\,\, \max_{H\in\F_{\Pihat_{j-1}}}\la S, H\ra -\lambda \beta \|H \|_{1,1} -\lambda (1-\beta) \|H\|_{1,1}^{\ast}\\
       &\Leftrightarrow \max_{H\in\F_{\Pihat_{j-1}}} \min_{Z_1\in\B_{p,1},
       \text{ } Z_2\in \B_{p,2}}\la S, H\ra - \lambda\beta \la H, Z_1\ra-\lambda (1-\beta) \la H,Z_2 \ra\\
       &\Leftrightarrow \min_{Z_1\in\B_{p,1},\text{ } Z_2\in \B_{p,2}} \max_{H\in\F_{\Pihat_{j-1}}} \la S -\lambda\beta  Z_1-\lambda (1-\beta) Z_2, H\ra. 
\end{align*}

By the Karush-Kuhn-Tucker (KKT) condition, a triplet  $\left(\Hhat,\Zhat_1,\Zhat_2\right)\in \mathcal{F}_{\Pihat_{j-1}}\times \B_{p,1} \times \B_{p,2}$ is optimal for problem \eqref{eq:model} if and only if 

\begin{align}
    (\Zhat_1)_{st}=\text{sign}\left(\Hhat_{st}\right), \quad \forall (s,t) \text{ with } \Hhat_{st} \neq 0,\label{eq:kkt1}\\
    (\Zhat_1)_{st}\in [0,1], \quad \forall (s,t) \text{ with } \Hhat_{st}=0,\label{eq:kkt2}\\
    (\Zhat_{2})^{k\ell}=\frac{w_{k\ell}\Hhat^{k\ell}}{\|\Hhat^{k\ell}\|_2}, \quad \forall (k,\ell) \text{ with } \|\Hhat^{k\ell}\|_2\neq 0,\label{eq:kkt3}\\
    \Zhat_2^{k\ell}\in \left\{U: \|U \|_2\leq w_{k\ell}\right\}, \quad  \forall (k,\ell) \text{ with } \|\Hhat^{k\ell} \|_2=0,\label{eq:kkt4} \\
    \Hhat=\arg\max_{H\in \F_{\Pihat_{j-1}}}\la S-\lambda \beta \Zhat_1-\lambda (1-\beta) \Zhat_2,H  \ra.\label{eq:kkt5}
\end{align}

To proceed, we first construct a sparse solution $\Htilde$:
\begin{equation}
    \label{eq:opt3}
    \Htilde = \arg\max_{H\in\F_{\Pihat_{j-1}},\,J_1(\text{diag}(H))\subseteq J_1} \la S, H\ra - \lambda\beta \| H \|_{1,1} - \lambda(1-\beta)\|H\|_{1,1}^{\ast}.
\end{equation}

Let $\left(\Ztilde_1, \Ztilde_2\right)$ be a pair of corresponding optimal dual variables. By Lemma \ref{lem:Q}, $\Htilde$ is a rank-1 projection matrix with non-zero elements supported on $J_1\times J_1$, and non-zero blocks supported on $J_2\times J_2$, i.e., $\Htilde_{st} = 0$ if $(s,t)\notin J_1\times J_1$, and $\|(\Htilde)^{k\ell}\|_2 = 0$ if $(k,\ell)\notin J_2\times J_2$.

Let $\bu$ be the leading eigenvector of 
$\Htilde$. By Lemma \ref{lem:Q}, there exists an $s_1\times s_1$ orthonormal matrix $Q_j$ and $\tilde{Q}_j = \text{blockdiag}(Q_j, 0_{(p-s_1)\times (p-s_1)})$ such that 
$$
[\bvhat_{1},\ldots,\bvhat_{j-1},\bu]  = 
\tilde{Q}_j
[\bv_{1},\ldots,\bv_{j-1}, \bv_j],
$$
and
$$
\|Q_j-I\|_2 \leq \sum_{k=1}^j \frac{4[2(\lambda \beta s_1+\lambda (1-\beta) s_2 w_1) + 2\gamma_k \epsilon_{k-1} + \gamma_1 \epsilon_{k-1}^2]}{\gamma_k-\gamma_{k+1}}.
$$

Denote $J_{1,\ast}$ as the collection of elementwise indices for nonzero groups $J_2$. Figure \ref{fig:sparse} gives an illustration of the sparsity structure in $\Sigma_1$.

\begin{figure}[htp]
           \centering            \includegraphics[width=9cm,height=9cm]{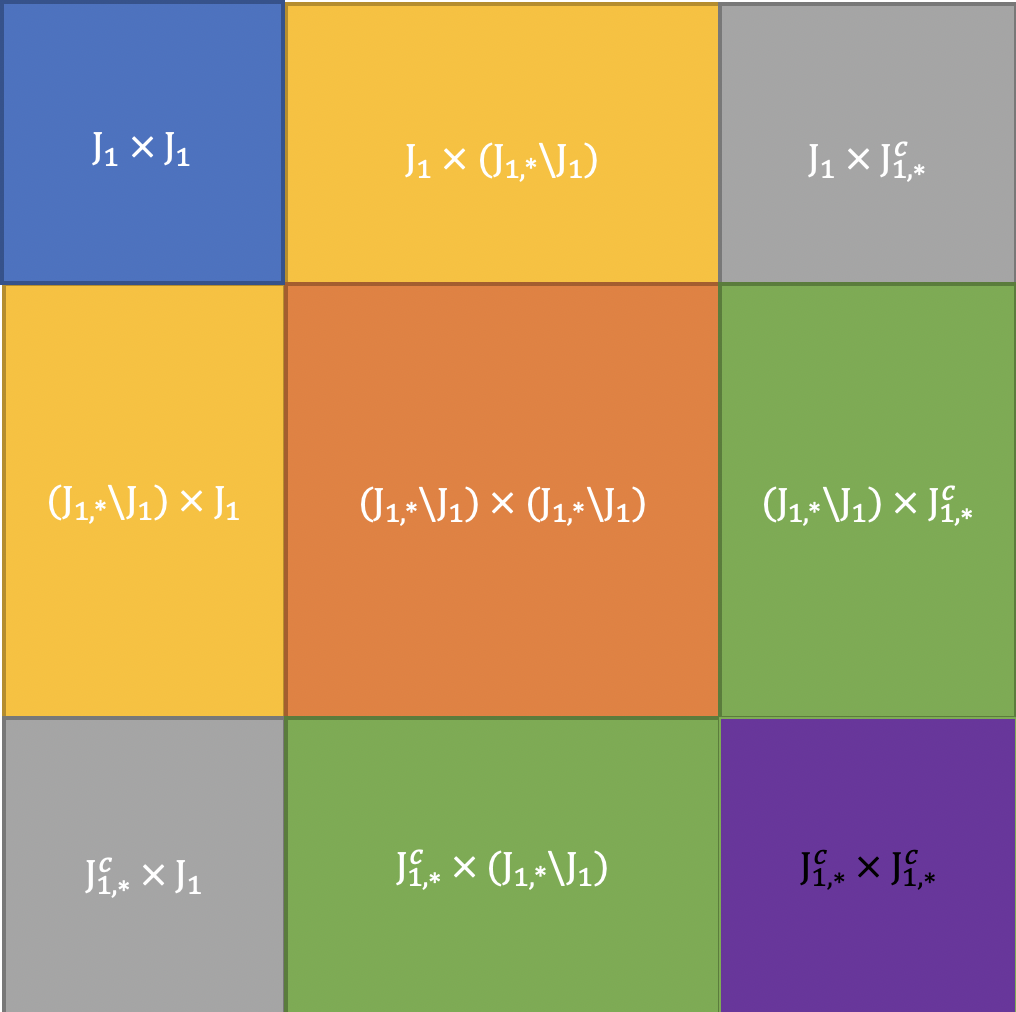}
            \caption{A schematic plot illustrating the blockwise partition of $\Sigma_1$ based on index sets $J_1$ and $J_{1,\ast}$. The upper left block colored blue and labelled by $J_1\times J_1$ includes all nonzero entries in the signal covariance matrix, since $J_1$ is the union of the elementwise support for eigenvectors $\bv_i, 1\leq i \leq r$. The orange block labelled by $(J_{1,\ast}\setminus J_1)\times (J_{1,\ast}\setminus J_1)$ contains zero entries whose columnwise and rowwise indices at the block level belong to nonzero groups $J_2$.
            The bottom right diagonal block corresponds to the entries with indices belonging to $J_{1,\ast}^c\times J_{1,\ast}^c$.}
            \label{fig:sparse}
\end{figure}

Now we define a modified primal-dual triplet $(\Hhat,\Zhat_1,\Zhat_2)$ as follows: 
\newpage
\begin{equation}\label{eq:mod}
    \begin{split}
         \Hhat &=\Htilde,\\
    (\Zhat_1)_{J_1J_1}&=(\Ztilde_1)_{J_1J_1},\\
     (\Zhat_2)^{J_2J_2}&= (\Ztilde_2)^{J_2J_2},\\
      (\Zhat_1)_{st}&=(\Zhat_1)_{ts} = \frac{1}{\lambda\beta}\left\{S_{st}-\la (Q_j)_{s \bdot}, (\Sigma_1)_{J_1t} \ra \right\},        \quad (s,t) \in J_1\times (J_{1,\ast}\setminus J_1),\\
           (\Zhat_1)_{st}&=\frac{1}{\lambda\beta}E_{st},
      \quad (s,t) \in (J_{1,\ast}\setminus J_1)\times  (J_{1,\ast}\setminus J_1),\\
      (\Zhat_1)_{st}&=(\Zhat_1)_{ts}=\frac{1}{\lambda}\left\{S_{st}-\la (Q_j)_{s \bdot}, (\Sigma_1)_{J_{1}t} \ra\right\}, \quad (s,t) \in J_1\times J_{1,\ast}^c,\\
      (\Zhat_2)_{st}&=(\Zhat_2)_{ts}=\frac{1}{\lambda}\left\{S_{st}-\la (Q_j)_{s \bdot}, (\Sigma_1)_{J_{1}t} \ra\right\},\quad (s,t) \in J_1\times J_{1,\ast}^c,\\
       (\Zhat_1)_{st}&=(\Zhat_1)_{ts}=\frac{1}{\lambda}E_{st}, \quad (s,t) \in J_1^c \times J_{1,\ast}^c,\\
      (\Zhat_2)_{st}&=(\Zhat_2)_{ts}=\frac{1}{\lambda}E_{st},\quad (s,t) \in J_1^c \times J_{1,\ast}^c.
      \end{split}
\end{equation}

We now check that $(\Hhat, \Zhat_1,\Zhat_2)$ is feasible and satisfies the KKT conditions \eqref{eq:kkt1} -\eqref{eq:kkt5}. 

{\it Feasibility of $\Hhat$.} Obviously $\Hhat$ is feasible.

{\it Feasibility of $\Zhat_2$.} 
The dual $\Zhat_2$ is feasible if $\|(\Zhat_2)^{k\ell}\|_2\leq  w_{k\ell}$
for all $(k,\ell)\in (J_2\times J_2^c)\cup(J_2^c\times J_2^c)$. Suppose first $(k,\ell)\in J_2\times J_2^c $, the elementwise entries of block $(\Zhat_2)^{k\ell}$ have indices $(s,t)\in J_{1,\ast}\times J_{1,\ast}^c$. By the definition of $Z_2$ in  \eqref{eq:mod}, we can equivalently write $(\Zhat_2)^{k\ell}$ as
\begin{equation*}
    (\Zhat_2)^{k\ell}=\frac{1}{\lambda}\left\{S^{k\ell}-(\bar{Q}_j)^{k\bdot}(\Sigma_1)^{J_2\ell}\right\},
\end{equation*}
where $\bar{Q}_j = \text{blockdiag}(Q_j, 0_{\{(|J_{1,\ast}|-s_1 )\times (|J_{1,\ast}| -s_1)\}})$.
It is easy to derive that
\begin{align*}
   \la (\bar{Q}_j)_{s\bdot}, (\Sigma_1)_{J_{1,\ast}t}  \ra =  
  \begin{cases}
    \la  (Q_j)_{s\bdot}, (\Sigma_1)_{J_1t},    \ra & \forall s\in J_1, t\in J_{1,\ast}^c\\
    (\Sigma_1)_{st}, & \forall s\in J_{1,\ast}\setminus J_1, t\in J_{1,\ast}^c.
    \end{cases}
\end{align*}
Hence, we have
\begin{equation*}
    \begin{split}
 \|(\Zhat_2)^{k\ell}\|_2&=  \frac{1}{\lambda}\|S^{k\ell}-\ (\bar{Q}_j)^{k \bdot}(\Sigma_1)^{J_2\ell} \|_2 \\
 &= \frac{1}{\lambda}\left\{\|S^{k\ell}-(\Sigma_1)^{k\ell}+ (\Sigma_1)^{k\ell}-(\bar{Q}_j)^{k \bdot}(\Sigma_1)^{J_2\ell}\|_2 \right\}\\
 &=\frac{1}{\lambda}\left\{\|E^{k\ell}-(I- \bar{Q}_j)^{k\bdot} (\Sigma_1)^{J_2\ell}\|_2 \right\}\\
 &\leq \frac{1}{\lambda}\left\{\|E^{k\ell}\|_2+\|I- \bar{Q}_j \|_2 \|(\Sigma_1)^{J_2\ell}\|_2\right\}\\
 &= \frac{1}{\lambda} \left\{w_{k\ell}\frac{\| E^{k\ell}\|_2}{w_{k\ell}}    + \|I- Q_j \|_2 \sqrt{\sum_{k\in J_2}w_{k\ell}^2\frac{\|(\Sigma_1)^{k\ell} \|_2^2}{w_{k\ell}^2}}\right\} \\
 &\leq  \frac{1}{\lambda}\left\{w_1\|E\|_{\infty,\infty}^{\ast}+\|I- Q_j \|_2 w_1 \|(\Sigma_1)^{J_2^cJ_2}\|_{2,\infty}^{\ast}\right\}\\
 &\leq w_0\\
 &\leq w_{k\ell}.
    \end{split}
\end{equation*}
In the above derivation, we used the facts
$\|(\Sigma_1)^{J_2^cJ_2}\|_{2,\infty}^{\ast}=\max_{\ell\in J_2^c} \sqrt{\sum_{k\in J_2}\frac{\|(\Sigma_1)^{k\ell} \|_2^2}{w_{k\ell}^2}}$, and $\|I- \bar{Q}_j \|_2= \|I- Q_j \|_2$. The last inequality is satisfied because of condition \eqref{eq:cond3}.

Now suppose $(k,\ell)\in J_2^c\times J_2^c$. By condition \eqref{eq:cond3},
\begin{equation*}
\begin{split}
      \left\|(\Zhat_2)^{k\ell} \right\|_2&\leq \frac{1}{\lambda}\| E^{k\ell} \|_2\\
  &\leq \frac{1}{\lambda}w_1\| E\|_{\infty,\infty}^{\ast}\\
  &\leq w_0\\
  &\leq w_{k\ell},
  \end{split}
\end{equation*}
which establishes the feasibility of $\Zhat_2$.

{\it Feasibility of $\Zhat_1$.} The dual variable $\Zhat_1$ is feasible if $|(\Zhat_1)_{st}|\in [0,1]$ for all $(s,t)\in (J_1\times J_1^c)\cup(J_1^c\times J_1^c)$.
First,  consider $(s,t)\in J_1\times (J_{1,\ast}\setminus J_1)$. We have
\begin{align*}
    |(\Zhat_1)_{st}|&=\frac{1}{\lambda\beta }|S_{st}-\la(Q_j)_{s\bdot}, (\Sigma_1)_{J_1t}\ra|\\
    &\leq \frac{1}{\lambda\beta}\left\{|S_{st}-(\Sigma_1)_{st}|+ |(\Sigma_1)_{st}-\la (Q_j)_{s\bdot}, (\Sigma_1)_{J_1t}\ra|\right\}\\
    &\leq \frac{1}{\lambda\beta}\left\{\| E \|_{\infty,\infty} + \|I-Q_j \|_2 \|(\Sigma_1)_{(J_{1,\ast}\setminus J_1)J_1} \|_{2,\infty}\right\}\\
    &\leq \frac{1}{\lambda\beta}\left\{\| E \|_{\infty,\infty} + \|I-Q_j \|_2 \|(\Sigma_1)_{J_1^cJ_1} \|_{2,\infty}\right\}\\
    &\leq 1,
\end{align*}
where $\|(\Sigma_1)_{J_1^cJ_1} \|_{2,\infty}=\max_{t\in J_1^c}\|(\Sigma_1)_{J_1t} \|_2$. The last inequality holds true because of condition \eqref{eq:cond4}. Similar derivation follows for $(s,t )\in J_1\times J_{1,\ast}^c$.


Next consider $(s,t)\in (J_{1,\ast}\setminus J_1)\times  (J_{1,\ast}\setminus J_1)$. Condition \eqref{eq:cond4} implies that
\begin{equation*}
    \begin{split}
         |(\Zhat_1)_{st}|&=\frac{1}{\lambda\beta}|E_{st}|\\
    &\leq \frac{1}{\lambda\beta}\|E \|_{\infty,\infty}\\
    &\leq 1.
    \end{split}
\end{equation*}
Similar derivation follows for $(s, t)\in J_1^c\times J_{1,\ast}^c$,
which concludes the proof of the feasibility of $\Zhat_1$.

{\it KKT conditions \eqref{eq:kkt1} and \eqref{eq:kkt3}.} Obvious by construction.

{\it KKT condition \eqref{eq:kkt2}. } For $(s,t)\in J_1\times J_1$, \eqref{eq:kkt2} holds for $(\Hhat,\Zhat_1)$ because the same holds for $(\Htilde,\Ztilde_1)$. For $(s,t)\notin J_1\times J_1$, $\|\Hhat_{st}\|_2 = 0$ and then \eqref{eq:kkt2} follows since $\Zhat_1$ is feasible.

{\it KKT condition \eqref{eq:kkt4}. } For $(k,\ell)\in J_2\times J_2$, \eqref{eq:kkt4} holds for $(\Hhat,\Zhat_2)$ because the same holds for $(\Htilde,\Ztilde_2)$. For $(k,\ell)\notin J_2\times J_2$, $\|\Hhat^{k\ell}\|_2 = 0$ and then \eqref{eq:kkt4} follows since $\Zhat_2$ is feasible. 

{\it KKT condition \eqref{eq:kkt5}.} By Lemma \ref{lem:ftp_prop},
it suffices to show that $\bu$ is the leading eigenvector of 
\begin{equation*}
\begin{split}
\tilde{\Sigma} &:= (I-\Pihat_{j-1}) (S-\lambda\beta\Zhat_1-\lambda(1-\beta)\Zhat_2)(I-\Pihat_{j-1})
\\
&=(I-\Pihat_{j-1})\left(\begin{array}{cc} 
S_{J_1J_1} - \lambda \beta (\Ztilde_1)_{J_1J_1}-\lambda(1-\beta) (\Ztilde_2)_{J_1J_1},&Q_j(\Sigma_1)_{J_1J_1^c}\\
(\Sigma_1)_{J_1^c J_1}Q_j^{\intercal},& (\Sigma_1)_{J_1^c J_1^c}
\end{array}\right)(I-\Pihat_{j-1}).
\end{split}
\end{equation*}
To simplify notation, let $D= I-\Pihat_{j-1}$. Note that all non-zero blocks of $D$ are in $D_{J_1J_1}$.
We have
\begin{align*}
    \tilde{\Sigma}
&= D\left(\begin{array}{cc} 
S_{J_1J_1} - \lambda\beta (\Ztilde_1)_{J_1J_1}-\lambda (1-\beta)(\Ztilde_2)_{J_1J_1} - Q_j(\Sigma_1)_{J_1 J_1}Q_j^{\intercal},&0\\
0,& 0
\end{array}\right)D\\
&\quad\,\,+ 
D\left(\begin{array}{cc} 
Q_j (\Sigma_1)_{J_1 J_1}Q_j^{\intercal},&Q_j (\Sigma_1)_{J_1J_1^c}\\
(\Sigma_1)_{J_1^c J_1}Q_j^{\intercal},&(\Sigma_1)_{J_1^c J_1^c}
\end{array}\right)D\\
&= \left(\begin{array}{cc} 
D_{J_1J_1}(S_{J_1J_1} - \lambda\beta (\Ztilde_1)_{J_1J_1} - \lambda (1-\beta) (\Ztilde_2)_{J_1J_1} -Q_j(\Sigma_1)_{J_1 J_1}Q_j^{\intercal})D_{J_1J_1},&0\\
0,& 0
\end{array}\right)\\
&\quad\,\,+ 
D
\left(
\begin{array}{cc}
    Q_j, & 0 \\
     0,& I 
\end{array}
\right)
\Sigma_1
\left(
\begin{array}{cc}
    Q_j^{\t}, & 0 \\
     0,& I 
\end{array}
\right)D\\
&=\text{``noise''} + \text{``signal''}.
\end{align*}


Let $\tilde{Q}_j = \text{blockdiag}(Q_j, I_{p - s_1})$.
Based on the spectral decomposition of $\Sigma_1$, the ``signal'' $D{Q}_j\Sigma_1{Q}_j^{\t}D$ can be decomposed into two parts,
\begin{align*}
D\tilde{Q}_j\Sigma_1\tilde{Q}_j^{\t}D=
D
\tilde{Q}_j
\left(\sum_{k=1}^j\gamma_k \bv_k \bv_k^{\intercal}\right)
\tilde{Q}_j^{\intercal}D+
D
\tilde{Q}_j
\left(\sum_{k=j+1}^r \gamma_k \bv_k \bv_k^{\t}\right)
\tilde{Q}_j^{\t}D.
\end{align*}
By Lemma \ref{lem:Q} and the definition of $D$, we have
\begin{equation*}
\begin{split}
   D\tilde{Q}_j\bv_k&=\begin{cases}
   0, & k<j,\\
   \bu, &k=j.
    \end{cases}
\end{split}
\end{equation*}
Thus,
\begin{equation*}
  D
\tilde{Q}_j
\left(\sum_{k=1}^j\gamma_k \bv_k \bv_k^{\t}\right)
\tilde{Q}_j^{\t}D= \gamma_j \bu \bu^{\t}.
\end{equation*}

For the second term, noting that
\begin{align*}
    \bu^{\t} D\tilde{Q}_j\bv_k=\bu^{\t}\tilde{Q}_j\bv_k=\bv_j^{\t}\bv_k=0 \quad \forall k \geq j+1, 
\end{align*}
showing that the second part is orthogonal to $\bu$. Since $D$ is a projection matrix, the leading eigenvalues of the second part is less or equal to $\gamma_{j+1}$. Hence, $\bu$ must be the leading eigenvector of the signal matrix with a spectral gap which is at least $\gamma_j-\gamma_{j+1}$.

Moreover,
$\bu_{J_1}$ is an eigenvector of $D_{J_1J_1}(S_{J_1J_1} - \lambda\beta (\Ztilde_1)_{J_1J_1} -\lambda(1-\beta) (\Ztilde_2)_{J_1J_1} - Q_j (\Sigma_1)_{J_1 J_1}Q_j^{\intercal})D_{J_1J_1}$. The reason is that Lemma \ref{lem:basic} shows that $\bu$ is the leading eigenvector of $D_{J_1J_1}(S_{J_1J_1} - \lambda_1 (\Ztilde_1)_{J_1J_1} -\lambda_2(\Ztilde_2)_{J_1J_1})D_{J_1J_1}$ and a simple calculation shows that $\bu_{J_1}$ is also an eigenvector of $D_{J_1J_1}Q_j (\Sigma_1)_{J_1 J_1}Q_j^{\intercal}D_{J_1J_1}$.

Thus, it suffices to show that the operator  of the noise matrix is smaller than $(\gamma_j - \gamma_{j+1})/2$.
We have
\begin{align*}
&\|S_{J_1J_1} -\lambda\beta (\Ztilde_1)_{J_1J_1} - \lambda (1-\beta) (\Ztilde_2)_{J_1J_1} - Q_j (\Sigma_1)_{J_1 J_1}Q_j^{\intercal}\|_{op}\\
\leq &
\|S_{J_1J_1} - \lambda\beta (\Ztilde_1)_{J_1J_1} - \lambda (1-\beta) (\Ztilde_2)_{J_1J_1}- (\Sigma_1)_{J_1J_1}\|_2
+ \|(\Sigma_1)_{J_1J_1} - Q_j (\Sigma_1)_{J_1 J_1}Q_j^{\intercal}\|_2\\
\leq &2(\lambda\beta s_1+\lambda(1-\beta) s_2 w_1) + 2\|(\Sigma_1)_{J_1J_1}\|_{op}\|Q_j-I\|_{2}\\
\leq &2(\lambda\beta s_1+\lambda(1-\beta) s_2 w_1) + 2\gamma_1 \times \sum_{k=1}^{j} \frac{4[2(\lambda\beta s_1+\lambda (1-\beta) s_2 w_1) + 2\gamma_j \epsilon_{k-1} + \gamma_1 \epsilon_{k-1}^2]}{\gamma_k-\gamma_{k+1}}.
\end{align*}
By condition~\eqref{eq:cond5}, proof of existence of a sparse solution is now complete.
\end{proof}

\begin{proof}[Uniqueness of solution]Consider the elastic net version of FPS:
\begin{equation}\label{eq:elastic}
\max_{H\in \F_{\Pihat_{j-1}}} \quad \la S, H\ra -\lambda\beta \|H\|_{1,1} - \lambda(1-\beta) \|H\|_{1,1}^{\ast} - \frac{\tau}{2} \|H\|_2^2.
\end{equation}
As the objective function is strongly concave, the solution of \eqref{eq:elastic} is unique. Consider the max-min and min-max forms using the dual variable $Z_1\in \B_{p,1} = \{Z\in \mathbb{R}^{p\times p}: Z = Z^T, \|Z\|_{\infty,\infty}^{\ast} \leq 1 \}$, $Z_2\in\B_{p,2} = \{Z\in \R^{p\times p}:
Z = Z^{\t}, \|Z^{k\ell}\|_2 \leq w_{k\ell},\forall (k,\ell)
\}$:
\begin{eqnarray*}
&&\min_{H\in\F_{\Pihat_{j-1}}}
\max_{Z_1\in\B_{p,1},Z_2\in\B_{p,2}} - \la S, H\ra + \lambda\beta \la H, Z_1\ra + \lambda (1-\beta) \la H, Z_2\ra + \frac{\tau}{2} \|H\|_2^2\\
&&\Leftrightarrow
\max_{Z_1\in \B_{p,1}, Z_2\in \B_{p,2}}\min_{H\in\F_{\Pihat_{j-1}}}
\frac{\tau}{2} \left\| H- \frac{1}{\tau}(S-\lambda\beta Z_1-\lambda (1-\beta) Z_2)\right\|_2^2 - \frac{1}{2\tau}\|S-\lambda\beta Z_1-\lambda (1-\beta) Z_2\|_2^2.
\end{eqnarray*}

By the Karush-Kuhn-Tucker (KKT) condition, a triplet $(\Hhat, \Zhat_1, \Zhat_2)\in\F_{\Pihat_{j-1}}\times \B_{p,1} \times \B_{p,2}$ is optimal for the above problem if and only if 
\begin{align*}
     (\Zhat_1)_{st}=\text{sign}\left(\Hhat_{st}\right), \quad \forall (s, t) \text{ with } \Hhat_{st} \neq 0,\\
    (\Zhat_1)_{st}\in [0,1], \quad \forall (s, t) \text{ with } \Hhat_{st}=0,\\
    (\Zhat_{2})^{k\ell}=\frac{w_{k\ell}\Hhat^{k\ell}}{\|\Hhat^{k\ell}\|_2}, \quad \forall (k,\ell) \text{ with } \|\Hhat^{k\ell}\|_2\neq 0,\\
    (\Zhat_2)^{k\ell}\in \left\{U: \|U \|_2\leq w_{k\ell}\right\}, \quad  \forall (k,\ell) \text{ with } \|\Hhat^{k\ell} \|_2=0,\\
     \Hhat= \mathcal{P}_{\F_{\Pihat_{j-1}}}\left\{\frac{1}{\tau}(S-\lambda\beta \Zhat_1-\lambda (1-\beta) \Zhat_2)\right\}.
\end{align*}
Denote $\Sigmatilde = (I - \Pihat_{j-1})(S-\lambda\beta \Zhat_1 -\lambda (1-\beta) \Zhat_2)(I - \Pihat_{j-1})$.
Let $(\Hhat, \Zhat_1, \Zhat_2)$ be the modified primal-dual triplet in \eqref{eq:mod}. We now show that this triplet is also optimal for the above optimization problem when $\tau$ is sufficiently small. By the existence proof,  $\lambda_1(\Sigmatilde) - \lambda_2 (\Sigmatilde) >0$. Then, part 3 of Lemma \ref{lem:ftp_prop} implies that when 
$$
0 < \tau \leq \lambda_1(\Sigmatilde) - \lambda_2(\Sigmatilde),
$$
we have
$$
\Hhat = 
\mathcal{P}_{\F^1}\left(\frac{1}{\tau}\Sigmatilde\right),
$$
which is equivalent to
$$
\Hhat= \mathcal{P}_{\F_{\Pihat_{j-1}}}\left\{\frac{1}{\tau}(S-\lambda\beta \Zhat_1-\lambda (1-\beta) \Zhat_2)\right\}
$$
by Lemma \ref{lem:basic}, and thus $(\Hhat, \Zhat_1, \Zhat_2)$ is indeed an optimal primal-dual triplet for the optimization problem \eqref{eq:elastic}.

To prove uniqueness of $\Hhat$ as a solution to \eqref{eq:model}, assume that there exists another solution $\Hhat'\in \F_{\Pihat_{j-1}}$ such that
$$
\la S, \Hhat\ra - \lambda\beta \|\Hhat\|_{1,1}- \lambda(1-\beta)\|\Hhat\|_{1,1}^{\ast}
=\la S, \Hhat'\ra -\lambda\beta\|\Hhat'\|_{1,1}- \lambda(1-\beta)\|\Hhat'\|_{1,1}^{\ast}.
$$
Because $\Hhat$ is the unique solution to \eqref{eq:elastic} for small $\tau>0$,
$\|\Hhat'\|_2^2 > \|\Hhat\|_2^2$.
In addition, by the existence proof, $\Hhat$ is of rank 1.
Hence,
$$
1 = \tr(\Hhat) =  \|\Hhat\|_2^2 < \|\Hhat'\|_2^2
\leq \tr(\Hhat') = 1
$$
and we have reached a contradiction.
\end{proof}

%% file: sections/Appendix/Technical_proofs/Propositions/Proposition3.tex
\begin{proof}[Proof of Proposition \ref{supp:element}]
For $j=1,\ldots, r$, let $\Delta_j=\Hhat_j-\bv_j\bv_j^{\t}$.
Note $(\bv_j\bv_j^{\t})_{ii} = v_{ji}^2$, then
\begin{equation*}
\begin{split}
    D_0 &\coloneqq \left\{i: v_{ji}=0, (\Hhat_j)_{ii}\geq t\right\} =\left\{i: (\bv_j\bv_j^{\t})_{ii} =0, (\Hhat_j)_{ii}\geq t\right\}\subseteq \left\{i: |(\Delta_j)_{ii}|\geq t\right\}, \\
    D_1 &\coloneqq \left\{i: v_{ji}^2 \geq \sqrt{2t},(\Hhat_j)_{ii} <t\right\} = \left\{i: (\bv_j\bv_j^{\t})_{ii}\geq 2t, (\Hhat_j)_{ii} <t\right\} \subseteq \left\{i: |(\Delta_j)_{ii}| \geq t\right\},
\end{split}
\end{equation*}
and $D_0\cap D_1=\emptyset$. By Markov's Inequality,
\begin{equation*}
|D_0|+|D_1|\leq \left|\left\{i:|(\Delta_j)_{ii}| \geq t\right\}\right|\leq \frac{1}{t^2}\sum_{i} |(\Delta_j)_{ii}|^2\leq \frac{\| \Delta_j\|_2^2}{t^2}.
\end{equation*}
\end{proof}

%% file: sections/Appendix/Technical_proofs/Propositions/Proposition4.tex
\begin{proof}[Proof of Proposition \ref{supp:block}]
Let $\Delta_j$ be defined as in the Proof of Proposition \ref{supp:element}.
Note $\|(\bv_j\bv_j^{\t})^{kk}\|_2 = \|\bv_j^k\|_2^2$, then
\begin{equation*}
\begin{split}
    D_0 &\coloneqq \left\{k: \|\bv_j^{k}\|_2=0, \|(\Hhat_j)^{kk}\|_2\geq t\right\}=\left\{k: \|(\bv_j\bv_j^{\t})^{kk}\|_2=0, \|(\Hhat_j)^{kk}\|_2\geq t\right\}\subseteq \left\{k: \|(\Delta_j)^{kk}\|_2\geq t\right\}, \\
    D_1 &\coloneqq \left\{k: \|\bv_j^{k}\|_2 \geq \sqrt{2t},\|(\Hhat_j)^{kk}\|_2 <t\right\} = \left\{k: \|(\bv_j\bv_j^{\t})^{kk}\|_2\geq 2t,\|(\Hhat_j)^{kk}\|_2 <t\right\} \subseteq \left\{k: \|(\Delta_j)^{kk}\|_2\geq t\right\},
\end{split}
\end{equation*}
and $D_0\cap D_1=\emptyset$. By Markov's Inequality,
\begin{equation*}
|D_0|+|D_1|\leq \left|\left\{k:\| (\Delta_j)^{kk}\|_2\geq t\right\}\right|\leq \frac{1}{t^2}\sum_{k}\| (\Delta_j)^{kk} \|_2^2\leq \frac{\| \Delta_j\|_2^2}{t^2}.
\end{equation*}

\end{proof}

%% file: sections/Appendix/Technical_proofs/Theorems/Theorem3.tex
\begin{proof}[Proof of Theorem \ref{thm:rate_block}]
For the proof we just need Proposition \ref{prop:1}, which we assume to hold for now.
Note that $|J_1(\bv_j)| \leq p$, $|J_1(\bvhat_j)|\leq p$,
$|J_2(\bv_j)| \leq I$, and $|J_2(\bvhat_j)|\leq I$. 
Assumption \ref{as:eta} says $\gamma_k = p \eta_k$.
Also $w_1 = p_0$. Thus,
$$
\tilde{a}_j \leq 2\left(\lambda \beta p +\lambda (1-\beta) p
 +  p \sum_{k=1}^{j-1}\eta_k \epsilon_k\right),
$$
and 
\begin{equation*}
    \begin{split}
 \tilde{b}_j = 4 \sqrt{2} \left\{\lambda\beta (j+1) p  +\lambda(1-\beta) (j+1) p
+ \frac{p}{2} \sum_{k=1}^{j-1}\eta_k
\epsilon_k \right\}
\epsilon_{j-1}+4\eta_j p \epsilon_{j-1}^2.
\end{split}
\end{equation*}

For $j=1$, we have $\tilde{a}_1  = O(\lambda p)$ and $\tilde{b}_1 =0$ by letting $\epsilon_0 = 0$.
Thus inequality \eqref{d_eb} in Proposition \ref{prop:1}
gives 
$$
\|\bvhat_1 - \bv_1\|_2 \leq 2\sqrt{2} \left(\frac{2\tilde{a}_1}{\gamma_1 - \gamma_2}\right) = O(\lambda).
$$

By induction,  we can choose $\epsilon_j = O(\lambda)$ and then show that  $\tilde{a}_j = O(\lambda p)$ and $\tilde{b}_j = O(\lambda^2 p)$, which then, by inequality \eqref{d_eb}, gives
$$
\|\bvhat_j - \bv_j\|_2 = O(\lambda).
$$

For Proposition \ref{prop:1} to hold, we need condition \eqref{eq:cond6}
to hold. Assumption \ref{as:prob_ieq} says $\|E\|_{\infty,\infty} = O_p(\sqrt{\log p/n})$, which also implies $\|E\|_{\infty,\infty}^{\ast} = O_p(\sqrt{\log p/n})$. Thus, if $\lambda \gtrapprox C \sqrt{\log p/n}$, then condition \eqref{eq:cond6} holds with high probability. Therefore, $$
\|\bvhat_j - \bv_j\|_2 = O_p(\lambda).
$$

The proof is complete.

\end{proof}

\begin{proof}[Proof of Theorem \ref{thm:recovery_I_strong}]
We first prove part $(i)$ via Corollary \ref{coro:block}. Notice that $\left\|(\Sigma_1)^{J_2^c J_2}\right\|_{2,\infty}^{\ast} = 0$ and thus it suffices to verify that  $\|E\|_{\infty,\infty}^{\ast}\leq \lambda$ holds with high probability and there exists a sufficiently large $n$ such that
$\gamma_j - \gamma_{j+1} - 4 (\lambda s_2 p_0 + \gamma_1 b_j)>0$ for $1\leq j \leq r$. For the first inequality, Assumption \ref{as:prob_ieq}
implies that $\|E\|_{\infty, \infty}^{\ast} = O_p(\sqrt{\log p/n})$. Hence if $\lambda \gtrapprox \sqrt{\log p/n}$, then $\|E\|_{\infty,\infty}^{\ast}\leq \lambda$ holds with high probability. As for the second inequality, first note
that 
$$
b_j = \sum_{k=1}^j \frac{8\lambda s_2 p_0 + 8\gamma_k \epsilon_{k-1} + 4\gamma_1 \epsilon_{k-1}^2}{\gamma_k - \gamma_{k+1}}.
$$
Thus by induction and the proof of Theorem \ref{thm:rate_block}, we can show that $b_j = O(\lambda) $.
So if $\lambda = O(\sqrt{\log p/n})$, then $\lambda = o(1)$. Thus the second inequality holds for any sufficiently large $n$.

We next prove part $(ii)$ via Proposition \ref{supp:block}. The proof of Theorem \ref{thm:rate_block}
also shows that $\|\hat{H}_j - \bv_j\bv_j^{\intercal}\|_2 = O_p(\lambda)$.
Hence,  the condition in Proposition \ref{supp:block} holds with high probability.

The proof is complete.
\end{proof}

%% file: sections/Appendix/Technical_proofs/Theorems/Theorem4.tex
\begin{proof}[Proof of Theorem \ref{thm:recovery_II_strong}]
We first prove part $(i)$ via Proposition \ref{prop:sparse}.  Notice that $\left\|(\Sigma_1)^{J_2^c J_2}\right\|_{2,\infty} = 0$ and also $\left\|(\Sigma_1)^{J_2^c J_2}\right\|_{2,\infty}^{\ast} = 0$, thus
Assumption \ref{as:prob_ieq} implies that both conditions \eqref{eq:cond3} and \eqref{eq:cond4} hold.

By induction and the proof of Theorem \ref{thm:rate_block}, we can show that $b_j = O(\lambda) $.
The assumption on $\lambda$ ensures that condition \eqref{eq:cond5} of Proposition \ref{prop:sparse} holds. Thus, the sparsity of the estimated eigenvectors follow.

We next prove part $(ii)$ via Proposition \ref{supp:element}. The proof of Theorem \ref{thm:rate_block}
also shows that $\|\hat{H}_j - \bv_j\bv_j^{\intercal}\|_2 = O_p(\lambda)$.
Hence, the condition in Proposition \ref{supp:element} holds with high probability.

The proof is complete.
\end{proof}

%% file: sections/Appendix/Technical_proofs/Theorems/Theorem5.tex
\begin{proof}[Proof of Theorem \ref{thm:rate_block_weak}]
For the proof we  need Proposition \ref{prop:1}, which we assume to hold for now and also Theorem \ref{thm:recovery_I_weak} part ($i$) from which we assume that
$J_2(\bvhat_k) \subseteq J_2$ for $k = 1,  \ldots, j-1$.

As $\beta=0$, we have
$$\tilde{a}_j = 2\lambda w_1 |J_2(\bv_j)| + 2\sum_{k=1}^{j-1} \gamma_k\epsilon_k,$$
and
$$
\tilde{b}_j = 4\sqrt{2}\left\{
\lambda w_1 \left(2 |J_2(\bv_j)| + \sum_{k=1}^{j-1} |J_2(\bvhat_k)|)\right)
+ \frac{1}{2}\sum_{k=1}^{j-1}\gamma_k \epsilon_k
\right\} \epsilon_{j-1} + 4\gamma_j \epsilon_{j-1}^2.
$$

Note that $w_1 = p_0$.
For $j=1$, we have $\tilde{a}_1 \leq 2\lambda p_0 s_2$ and $\tilde{b}_1 =0$ by letting $\epsilon_0 = 0$.
Thus inequality \eqref{d_eb} in Proposition \ref{prop:1}
gives 
$$
\|\bvhat_1 - \bv_1\|_2 \leq 2\sqrt{2} \left(\frac{2\tilde{a}_1}{\gamma_1 - \gamma_2}\right) = O(\lambda p_0 s_2).
$$

Suppose by induction that  $|J_2(\bvhat_k)|\leq s_2 $ (refer to part ($i$) of Theorem \ref{thm:recovery_I_weak}) and $\epsilon_k (k=1,\ldots, j-1)$ are chosen such that $\epsilon_k = O(p_0 s_2 \lambda)$.
Then we derive that
$$
\tilde{a}_j = O(\lambda p_0 s_2),
$$
and
$$
\tilde{b}_j = O( (\lambda p_0 s_2)^2).
$$
Then by inequality \eqref{d_eb}, we have
$$
\|\bvhat_j - \bv_j\|_2 = O(\lambda p_0 s_2).
$$

For Proposition \ref{prop:1} to hold, we just need condition \eqref{eq:cond6}
to hold. The same argument as in the proof of Theorem \ref{thm:rate_block} implies that condition \eqref{eq:cond6} holds with high probability and thus $$
\|\bvhat_j - \bv_j\|_2 = O_p(p_0 s_2 \sqrt{\log p/n}).
$$

The proof is complete.
\end{proof}

%% file: sections/Appendix/Technical_proofs/Theorems/Theorem6.tex
\begin{proof}[Proof of Theorem \ref{thm:recovery_I_weak}]
We first prove part $(i)$ via Corollary \ref{coro:block} and it suffices to verify that  $\|E\|_{\infty,\infty}^{\ast}\leq \lambda$ holds with high probability and there exists a sufficiently large $n$ such that
$\gamma_j - \gamma_{j+1} - 4 (\lambda s_2 p_0 + \gamma_1 b_j)>0$ for $1\leq j \leq r$. For the first inequality, Assumption \ref{as:prob_ieq}
implies that $\|E\|_{\infty, \infty}^{\ast} = O_p(\sqrt{\log p/n})$. Hence if $\lambda \gtrapprox  \sqrt{\log p/n}$, then $\|E\|_{\infty,\infty}^{\ast}\leq \lambda$ holds with high probability. As for the second inequality, first note
that 
$$
b_j = \sum_{k=1}^j \frac{8\lambda s_2 p_0 + 8\gamma_k \epsilon_{k-1} + 4\gamma_1 \epsilon_{k-1}^2}{\gamma_k - \gamma_{k+1}}.
$$
Thus by induction and the proof of Theorem \ref{thm:rate_block_weak}, we can show that $b_j = O(\lambda p_0 s_2) $.
So if $\lambda = O(\sqrt{\log p/n})$, by the assumption that $p_0 s_2 \sqrt{\log p/n} = o(1)$, the second inequality holds.

We next prove the support recovery part via Proposition \ref{supp:block}. The proof of Theorem \ref{thm:rate_block_weak}
also shows that $\|\hat{H}_j - \bv_j\bv_j^{\intercal}\|_2 = O_p(p_0 s_2 \sqrt{\log p/n})$.
Hence,  the condition in Proposition \ref{supp:block} holds with high probability.

The proof is complete.
\end{proof}

%% file: sections/Appendix/Technical_proofs/Theorems/Theorem7.tex
\begin{proof}[Proof of Theorem \ref{thm:rate_joint_weak}]
For the proof we  need Proposition \ref{prop:1}, which we assume to hold for now and also Theorem \ref{thm:recovery_II_weak} part ($i$) from which we assume that
$J_1(\bvhat_k) \subseteq J_1$ for $k = 1,  \ldots, j-1$.

Note that $w_1 = p_0$.
For $j=1$, we have $\tilde{a}_1  = O(\lambda\beta s_1 + \lambda(1-\beta) p_0 s_2)$ and $\tilde{b}_1 =0$ by letting $\epsilon_0 = 0$.
Thus inequality \eqref{d_eb} in Proposition \ref{prop:1}
gives 
$$
\|\bvhat_1 - \bv_1\|_2 \leq 2\sqrt{2} \left(\frac{2\tilde{a}_1}{\gamma_1 - \gamma_2}\right) = O(s_1\lambda).
$$
Recall that we assume $s_2 p_0/s_1 = O(1)$.

Suppose by induction that  $|J_1(\bvhat_k)|\leq s_1 $ (by part ($i$) of Theorem \ref{thm:recovery_II_weak}) and $\epsilon_k (k=1,\ldots, j-1)$ are chosen such that $\epsilon_k = O(s_1 \lambda)$.
Then we derive that
$$
\tilde{a}_j = O(s_1\lambda),
$$
and
$$
\tilde{b}_j = O( s_1^2 \lambda^2).
$$
Then by inequality \eqref{d_eb}, we have
$$
\|\bvhat_j - \bv_j\|_2 = O(s_1\lambda).
$$

For Proposition \ref{prop:1} to hold, we need condition \eqref{eq:cond6}
to hold. By Assumption \ref{as:prob_ieq}, if $\lambda\gtrapprox \sqrt{\log p/n}$, then condition \eqref{eq:cond6} holds with high probability and thus $$
\|\bvhat_j - \bv_j\|_2 = O_p(s_1\sqrt{\log p/n}).
$$

The proof is complete.
\end{proof}

%% file: sections/Appendix/Technical_proofs/Theorems/Theorem8.tex
\begin{proof}[Proof of Theorem \ref{thm:recovery_II_weak}]
We first prove part $(i)$ via Proposition \ref{prop:sparse}.
Assumption \ref{as:prob_ieq} implies that both conditions \eqref{eq:cond3} and \eqref{eq:cond4} hold.

By induction and the proof of Theorem \ref{thm:rate_joint_weak}, we can show that $b_j = O(s_1\lambda) $.
So by the assumption that $s_1 \sqrt{\log p/n} = o(1)$
and $\lambda =O(\sqrt{\log p/n})$, $b_j = o(1)$.
It follows that \eqref{eq:cond5} of Proposition \ref{prop:sparse} also holds. 

We next prove part $(ii)$ via Proposition \ref{supp:element}. The proof of Theorem \ref{thm:rate_joint_weak}
also shows that $\|\hat{H}_j - \bv_j\bv_j^{\intercal}\|_2 = O_p(s_1 \sqrt{\log p/n})$.
Hence,  the condition in Proposition \ref{supp:element} holds with high probability.

The proof is complete.
\end{proof}

%% file: sections/Appendix/Technical_proofs/Lemmas.tex
\begin{lemma}
\label{lem:vhat}
For $1\leq j\leq r$,
\begin{equation*}
\lVert  \hat{\bv}_j\hat{\bv}_j^{\intercal}-\bv_j \bv_j^{\intercal}\rVert_2\leq 2 \lVert  \bv_j\bv_j^{\intercal}-\hat{H}_j \rVert_2.
\end{equation*}
\end{lemma}
\begin{remark}
If $\hat{H}_j$ has rank 1, then $\hat{H}_j = \hat{\bv}_j \hat{\bv}_j^{\intercal}$.
\end{remark}
\begin{proof}[Proof of Lemma \ref{lem:vhat}]
By the triangle inequality, 
\begin{align*}
\lVert  \hat{\bv}_j\hat{\bv}_j^{\intercal}-\bv_j \bv_j^{\intercal}\rVert_2\leq \lVert \hat{\bv}_j\hat{\bv}_j^{\intercal}-\hat{H}_j  \rVert_2+ \lVert \hat{H}_j-\bv_j \bv_j^{\intercal}   \rVert_2.
\end{align*}
We derive that
\begin{align*}
\lVert \hat{\bv}_j\hat{\bv}_j^{\intercal}-\hat{H}_j  \rVert_2^2&=\operatorname{tr}(\hat{\bv}_j\hat{\bv}_j^{\intercal}\hat{\bv}_j\hat{\bv}_j^{\intercal})+\operatorname{tr}(\hat{H}_j\hat{H}_j^{\intercal})-2\langle \hat{\bv}_j\hat{\bv}_j^{\intercal}, \hat{H}_j  \rangle\\
&=1+\operatorname{tr}(\hat{H}_j\hat{H}_j^{\intercal})-2\hat{\bv}_j^{\intercal}\hat{H}_j\hat{\bv}_j.
\end{align*}
Note that $\hat{\bv}_j=\operatorname*{argmax}_{\lVert \bv \rVert_2=1} \bv^{\intercal}\hat{H}_j \bv$, we have $\hat{\bv}_j^{\intercal}\hat{H}_j\hat{\bv}_j\geq \bv_j^{\intercal}\hat{H}_j\bv_j$, thus 
\begin{align*}
\lVert \hat{\bv}_j\hat{\bv}_j^{\intercal}-\hat{H}_j  \rVert_2^2&\leq1+\operatorname{tr}(\hat{H}_j \hat{H}_j^{\intercal})-2\bv_j\hat{H}_j\bv_j^{\intercal}\\
&=\operatorname{tr}(\bv_j\bv_j^{\intercal}\bv_j\bv_j^{\intercal})+\operatorname{tr}(\hat{H}_j \hat{H}_j^{\intercal})-2\langle  \bv_j\bv_j^{\intercal},\hat{H}_j   \rangle\\
&=\lVert  \bv_j\bv_j^{\intercal}-\hat{H}_j \rVert_2^2.
\end{align*}
The proof is then complete.
\end{proof}

 \begin{lemma}\label{lem:ineq}
For any $\Delta, \Pi \in\R^{p\times p}$,
\begin{align*}
 \|\Delta \|_{1,1}-\|\Delta+\Pi \|_{1,1}+\|\Pi \|_{1,1}&\leq 2\sqrt{|J_1(\Pi)|}\|\Delta \|_2,\\
\|\Delta \|_{1,1}^{\ast}-\|\Delta+\Pi \|_{1,1}^{\ast}+\|\Pi \|_{1,1}^{\ast}&\leq 2\sqrt{\sum_{(k,\ell)\in J_2(\Pi)}w_{k\ell}^2}\|\Delta \|_2.
 \end{align*}
 \end{lemma}
\begin{proof} [Proof of Lemma \ref{lem:ineq}]
For simplicity, let $J_1 = J_1(\Pi)$
and $J_2 = J_2(\Pi)$.
Also define $J_1^c$ and $J_2^c$.
For an arbitrary matrix $B\in \R^{p\times p}$, 
$
\lVert B\rVert_{1,1}= \lVert B_{J_1} \rVert_{1,1}+\lVert B_{J_1^c} \rVert_{1,1}.
$
Thus,
\begin{align*}
&\lVert \Delta\rVert_{1,1}-\lVert \Pi+\Delta\rVert_{1,1}+\lVert \Pi\rVert_{1,1}\\
&=\lVert \Delta_{J_1}\rVert_{1,1}+\lVert \Delta_{J_1^c}\rVert_{1,1}-(\lVert \Delta_{J_1}+\Pi_{J_1}\rVert_{1,1}+\lVert \Delta_{J_1^c}\rVert_{1,1})+\lVert \Pi_{J_1} \rVert_{1,1}\\
&=\lVert \Delta_{J_1}\rVert_{1,1}-\lVert \Delta_{J_1}+\Pi_{J_1}\rVert_{1,1}+\lVert \Pi_{J_1} \rVert_{1,1}\\
&\leq 2\lVert \Delta_{J_1}\rVert_{1,1}
\leq 2\sqrt{|J_1|}\lVert \Delta_{J_1} \rVert_2
\leq 2 \sqrt{|J_1|}\lVert \Delta \rVert_2.
\end{align*}

Next, for an arbitrary  matrix $B\in \R^{p\times p}$,
$$
  \|B\|_{1,1}^{\ast}=  \sum_{k=1}^{I}\sum_{\ell=1}^{I}w_{k\ell}\lVert B^{k\ell} \rVert_2=\sum_{(k,\ell)\in J_2}w_{k\ell}\lVert B^{k\ell}\rVert_2+\sum_{(k,\ell)\in J_2^c}w_{k\ell}\lVert B^{k\ell}\rVert_2.
$$
Then,
\begin{align*}
&\|\Delta \|_{1,1}^{\ast}-\|\Delta+\Pi \|_{1,1}^{\ast}+\|\Pi \|_{1,1}^{\ast}\\
&=\sum_{k=1}^{I}\sum_{l=1}^{I}w_{k\ell}(\lVert \Delta^{k\ell}\rVert_2-\lVert \Delta^{k\ell} +\Pi^{k\ell} \rVert_2+\lVert \Pi^{k\ell}\rVert_2)\\
 &=\sum_{(k,\ell)\in J_2}w_{k\ell}(\lVert \Delta^{k\ell}\rVert_2-\lVert \Delta^{k\ell} +\Pi^{k\ell} \rVert_2+\lVert \Pi^{k\ell}\rVert_2)
 +\sum_{(k,\ell)\in J_2^c}w_{k\ell}(\lVert \Delta^{k\ell}\rVert_2-\lVert \Delta^{k\ell} +\Pi^{k\ell} \rVert_2+\lVert \Pi^{k\ell}\rVert_2)\\
 &=\sum_{(k,\ell)\in J_2}w_{k\ell}(\lVert \Delta^{k\ell}\rVert_2-\lVert \Delta^{k\ell} +\Pi^{k\ell} \rVert_2+\lVert \Pi^{k\ell}\rVert_2)
 +\sum_{(k,\ell)\in J_2^c}w_k(\lVert \Delta^{k\ell}\rVert_2-\lVert \Delta^{k\ell}  \rVert_2)\\
 &=\sum_{(k,\ell)\in J_2}w_{k\ell}(\lVert \Delta^{k\ell}\rVert_2-\lVert \Delta^{k\ell} +\Pi^{k\ell} \rVert_2+\lVert \Pi^{k\ell}\rVert_2)\\
 &\leq 2\sum_{(k,\ell)\in J_2}w_{k\ell}\lVert \Delta^{k\ell}\rVert_2\\
 &\leq 2\sqrt{\sum_{(k,\ell)\in J_2}w_{k\ell}^2}\sqrt{\sum_{(k,\ell)\in J_2}\lVert \Delta^{k\ell}\rVert_2^2}\\
 &=2\sqrt{\sum_{(k,\ell)\in J_2}w_{k\ell}^2}\lVert \Delta \rVert_2.
    \end{align*}
The proof is now complete.
\end{proof}

\begin{lemma}\label{lem:Q}
Let $\Htilde$ be the solution to \eqref{eq:opt3}.
Then $\Htilde$ is  unique and has rank 1.
Let $\bu$ be the eigenvector of $\Htilde$. Then,
there exists an $s_1\times s_1$ orthonormal matrix $Q_j$ such that
    \begin{align*}
        [\bvhat_{1},\ldots,\bvhat_{j-1},\bu] & = \left(\begin{array}{cc}Q_j,&0\\0,&I\end{array}\right)[\bv_{1},\ldots,\bv_{j-1},\bv_j],\\
        \|Q_j-I\|_2 &\leq 
        \sum_{k=1}^{j} \frac{4[2(\lambda\beta s_1+\lambda (1-\beta) s_2 w_1) + 2\gamma_k \epsilon_{k-1} + \gamma_1 \epsilon_{k-1}^2]}{\gamma_k-\gamma_{k+1}},
    \end{align*}
    where $\epsilon_{k-1}$ is any constant such that  $\|\Pihat_{k-1} - \Pi_{k-1}\|_2\leq \epsilon_{k-1}$.
\end{lemma}
\begin{proof}[Proof of Lemma \ref{lem:Q}]
Let $\tilde{\Sigma} = (I-\Pihat_{j-1})(S - \lambda\beta \Ztilde_1-\lambda (1-\beta) \Ztilde_2)(I-\Pihat_{j-1})$. Then as we may assume that $J_1(\bvhat_k)\subseteq J_1$ for $k\leq j-1$, it can be shown that $$\tilde{\Sigma}_{J_1J_1} = (I- (\Pihat_{j-1})_{J_1J_1})\left(S - \lambda\beta \Ztilde_1-\lambda (1-\beta) \Ztilde_2\right)_{J_1J_1}(I-(\Pihat_{j-1})_{J_1J_1}).$$ 

We first show that $\Htilde_{J_1J_1}$ is unique and has rank 1, which implies that $\Htilde$ is also unique and has rank 1.
Note that $(\hat{\Pi}_{j-1})_{J_1 J_1}$ remains a projection matrix. Then by Lemma \ref{lem:basic},  $\Htilde_{J_1J_1}$ maximizes $\la \tilde{\Sigma}_{J_1J_1}, H\ra$ over all 
$H\in \F^1$. Here with slight abuse of notation, $\mathcal{F}^1$ contains $s_1\times s_1$ matrices.
Let $\Sigmahat_j = (I-\Pihat_{j-1})\Sigma_1 (I-\Pihat_{j-1})$.
Then,  $\tilde{\Sigma}_{J_1J_1} - (\Sigmahat_{j})_{J_1J_1} = (I - (\Pihat_{j-1})_{J_1J_1}) (E-\lambda\beta\Ztilde_1-\lambda (1-\beta) \Ztilde_2)_{J_1J_1} (I - (\Pihat_{j-1})_{J_1J_1})$.

Thus,
\begin{equation}\label{eq:lem3.1}
\begin{split}
\|\tilde{\Sigma}_{J_1J_1} - (\Sigmahat_j)_{J_1J_1}\|_2
&\leq 
 \|(E-\lambda\beta\Ztilde_1-\lambda (1-\beta) \Ztilde_2)_{J_1J_1}\|_2\\
 & \leq \|E_{J_1J_1} \|_2+ \lambda\beta \| (\Ztilde_1)_{J_1J_1} \|_2 +\lambda (1-\beta)  \|(\Ztilde_2)_{J_1J_1} \|_2
 \end{split}
\end{equation}

On one hand, 
\begin{equation*}
\| E_{J_1J_1}\|_2 \leq s_1 \|E \|_{\infty,\infty}\leq \lambda\beta  s_1.
\end{equation*}
On the other hand,
\begin{equation*}
\begin{split}
\| E_{J_1J_1}\|_2 &\leq \|E^{J_2J_2} \|_2\\
&=\sqrt{\sum_{(k,\ell)\in J_2\times J_2}\|E^{k\ell} \|_2^2}\\
&=\sqrt{\sum_{(k,\ell)\in J_2\times J_2}\left(\frac{\|E^{k\ell} \|_2}{w_{k\ell}}
\right)^2w_{k\ell}^2}\\
&\leq s_2 w_1 \|E \|_{\infty,\infty}^{\ast}\\
&\leq s_2 w_1 \lambda (1-\beta).
\end{split}
\end{equation*}
Thus, we have
$$
\| E_{J_1J_1}\|_2\leq \min(\lambda\beta s_1, \lambda (1-\beta) s_2 w_1).
$$

Next, by the feasibility of $\Ztilde_1$ and $\Ztilde_2$, we have $ \| (\Ztilde_1)_{J_1J_1} \|_2 \leq s_1
 $ and
\begin{equation*}
\begin{split}
      \| (\Ztilde_2)_{J_1J_1} \|_2 & \leq \| (\Ztilde_{2})^{J_2J_2}\|_2 \\
    &=\sqrt{\sum_{(k,\ell)\in J_2\times J_2}\|\Ztilde_2^{k\ell} \|_2^2}\\
    &\leq \sqrt{\sum_{(k,\ell)\in J_2\times J_2}w_{k\ell}^2}\\
    &\leq s_2w_1.
\end{split}
\end{equation*}

Hence, \eqref{eq:lem3.1} becomes
\begin{equation*}
\begin{split}
\|\tilde{\Sigma}_{J_1J_1} - (\Sigmahat_j)_{J_1J_1}\|_2 &\leq \min (\lambda\beta s_1,\lambda (1-\beta) s_2w_1 )+\lambda\beta s_1+\lambda (1-\beta) s_2w_1\\
&\leq 2(\lambda\beta s_1+\lambda (1-\beta) s_2w_1).
\end{split}
\end{equation*}

Let $\Sigma_j = (I-\Pi_{j-1})\Sigma_1 (I-\Pi_{j-1})$. Then,
$\Sigma_j = \sum_{k=j}^r \gamma_k \bv_j\bv_k^{\t}$.
By Lemma \ref{lem:proj}, 
$$\|\Sigmahat_j - \Sigma_j\|_2\leq 2\gamma_j \|\Pihat_{j-1} - \Pi_{j-1}\|_2 + \gamma_1 \|\Pihat_{j-1} - \Pi_{j-1}\|_2^2
= 2\gamma_j \epsilon_{j-1} + \gamma_1 \epsilon_{j-1}^2.$$ 
Hence,
$\|(\Sigmahat_j)_{J_1J_1} - (\Sigma_j)_{J_1J_1}\|_2\leq 
2\gamma_j \epsilon_{j-1} + \gamma_1 \epsilon_{j-1}^2.$ It follows that
$$
 \|\tilde{\Sigma}_{J_1J_1} - (\Sigma_j)_{J_1J_1}\|_2
 \leq 2(\lambda\beta s_1+\lambda (1-\beta) s_2 w_1) + 2\gamma_j \epsilon_{j-1} + \gamma_1 \epsilon_{j-1}^2.
$$

We have $\gamma_1((\Sigma_j)_{J_1J_1}) =\gamma_1(\Sigma_j) =  \gamma_j$ and $\gamma_2((\Sigma_j)_{J_1J_1}) \leq \gamma_2(\Sigma_j) {\leq} \gamma_{j+1}$.
Thus, by Weyl's inequality,
\begin{align*}
    \gamma_1(\tilde{\Sigma}_{J_1J_1})
    - \gamma_2(\tilde{\Sigma}_{J_1J_1})&\ge
    {\gamma_1((\Sigma_j)_{J_1J_1}) - \gamma_2((\Sigma_j)_{J_1J_1})}-2\|(\tilde{\Sigma})_{J_1J_1}-{(\Sigma_1)}_{J_1J_1}\|_{op}\\
&\geq \gamma_j - \gamma_{j+1} - 4(\lambda\beta s_1+\lambda (1-\beta) s_2w_1) - 4\gamma_j \epsilon_{j-1}- 2\gamma_1 \epsilon_{j-1}^2>0.
\end{align*}
It follows that by part 1 of Lemma 1 in \cite{lei2015sparsistency}, $(\Htilde)_{J_1 J_1}$ is unique and has rank 1. 

Next,
by choosing the sign of $\bu$ such that $\la \bu, \bv_j\ra \geq 0$,
we have 
$$
\|\bu - \bv_j\|_2 \leq \|\bu\bu^{\t} - \bv_j\bv_j^{\t}\|_2.
$$
By Lemma \ref{lem:curv},
$$
\|\bu\bu^{\t} - \bv_j\bv_j^{\t}\|_2
\leq \frac{2}{\gamma_j - \gamma_{j+1}}
\|(\tilde{\Sigma})_{J_1J_1}-(\Sigma_j)_{J_1J_1}\|_2
\leq \frac{2[2(\lambda\beta s_1+\lambda (1-\beta) s_2 w_1) + 2\gamma_j \epsilon_{j-1} + \gamma_1 \epsilon_{j-1}^2]}{\gamma_j-\gamma_{j+1}}.
$$
So we have derived that
$$
\|\bu - \bv_j\|_2 \leq
\frac{2[2(\lambda\beta s_1+\lambda (1-\beta) s_2 w_1) + 2\gamma_j \epsilon_{j-1} + \gamma_1 \epsilon_{j-1}^2]}{\gamma_j-\gamma_{j+1}}.
$$

By definition, $\bu$ is orthogonal to $\bvhat_k$ for $k\leq j-1$.  Let
$\bVhat$ be an $s_1\times s_1$ orthonormal matrix such that its first $j$ columns
are $[(\bvhat_1)_{J_1},\ldots, (\bvhat_{j-1})_{J_1},\bu_{J_1}]$. Also let $V$ be an $s_1\times s_1$ orthonormal matrix such that its first $j$ columns are
$ [(\bv_1)_{J_1},\ldots, (\bv_{j-1})_{J_1},(\bv_j)_{J_1}]$. In particular, the last $(s_1-j)$ columns of $\bVhat$ and $V$ can be chosen such that 
$$
\|V - \bVhat \|_2 \leq \sum_{k=1}^{j} \frac{4[2(\lambda\beta s_1+\lambda (1-\beta) s_2 w_1) + 2\gamma_k \epsilon_{k-1} + \gamma_1 \epsilon_{k-1}^2]}{\gamma_k-\gamma_{k+1}}.
$$

Let $Q_j = \bVhat V^{\t}$. Then, $Q_j[(\bv_1)_{J_1},\ldots, (\bv_j)_{J_1}] = [(\bvhat_1)_{J_1},\ldots, \bvhat_{j-1}, \bu_{J_1}]$
and 
$$
\|I - Q_j\|_2  = \|(V-\bVhat)V^{\t}\|_2 = \|V - \bVhat \|_2
$$
and the proof is complete.
\end{proof}

\begin{lemma}
\label{lem:basic}
Let $A$ be a symmetric and positive semi-definite matrix
and $B=VV^{\t}$ be a projection matrix of the same dimension. Denote by $\gamma_1$ and $\gamma_2$ the two largest eigenvalues of $(I-B)A(I-B)$ with associated eigenvectors $\bv_1$ and $\bv_2$.
Then, $$\max_{H\in\F_{B}}\la A, H\ra = \max_{H\in \F^1} \la (I-B)A(I-B), H\ra = \gamma_1$$ and the maximum of both is achieved by $\bv_1\bv_1^{\t}$. Moreover, the maximizer is unique if and only if $\gamma_1 > \gamma_2$.
\end{lemma}
\begin{proof}
Let $U$ be a matrix that forms an orthogonal complement basis of $V$. First, we prove for $ \forall H\in \F_{B}$, there exists $G\in \F^1$, such that $H=UGU^{\t}$. Suppose that $H$ has spectral decomposition $H=\tilde{U}D\tilde{U}^{\t}$. 
By the orthogonality constraint
\begin{equation*}
\la H,B \ra = 0,
\end{equation*}
$\tilde{U}$ should satisfy $\tilde{U}^{\t}V = 0$. Or equivalently, $\tilde{U}$ lies in the subspace spanned by the columns of $U$. Namely, $\exists \tilde{A}, \text{s.t. } \tilde{U} = U\tilde{A}$. Hence,
\begin{equation}
\label{eq:HG}
\begin{split}
    H = \tilde{U} D \tilde{U}^{\t} = U (\tilde{A}D\tilde{A}^{\t}) U^{\t}.
\end{split}
\end{equation}
Denote $G=\tilde{A}D\tilde{A}^{\t}$, then \eqref{eq:HG} implies that $H$ and $G$ has the same set of eigenvalues. According to the fact that $H\in \F_B\subset \F^1$, we have $ 0 \preceq G\preceq I_p \text{ and } \tr(G) = 1$, which suggests that $G\in \F^1$. In the other way around, it is obvious that $\forall G \in \F^1$, $UGU^{\t}\in \F_B$. Hence,
\begin{equation*}
    \max_{H\in\F_{B}}\la A, H\ra = \max_{G\in \F^1} \la A, UGU^{\t}\ra = \max_{G\in \F^1} \la U^{\t}AU, G\ra. 
\end{equation*}
By Ky Fan's maximum principle (\cite{fan1949theorem}),
\begin{equation*}
    \max_{G\in \F^1} \la U^{\t}AU, G\ra = \la U^{\t}AU, \boldsymbol{g}_1\boldsymbol{g}_1^{\t} \ra = \tilde{\gamma}_1,
\end{equation*}
where $(\tilde{\gamma}_1,\boldsymbol{g}_1)$ is the leading eigenvalue-eigenvector pair of $U^{\t}AU$.
Similarly, 
\begin{equation*}
   \max_{H\in \F^1} \la (I-B)A(I-B), H\ra =  \la (I-B)A(I-B), \bv_1\bv_1^{\t}\ra = \gamma_1,
\end{equation*}
where $(\gamma_1,\bv_1)$ is the leading eigenvalue-eigenvector pair of $(I-B)A(I-B)$.

Since $I-B=UU^{\t}$, we have
\begin{equation*}
 (I-B)A(I-B)=U(U^{\t}AU)U^{\t},
\end{equation*}
which implies $\gamma_1=\tilde{\gamma}_1$ and  $\bv_1=U\boldsymbol{g}_1$, due to the fact $U^{\t}U=I$.
Therefore, the maximizer of 
$\max_{G\in \F^1} \la A, UGU^{\t}\ra$ is
\begin{equation*}
    U(\boldsymbol{g}_1\boldsymbol{g}_1^{\t})U^{\t} = U\boldsymbol{g}_1 (U\boldsymbol{g}_1)^{\t} = \bv_1\bv_1^{\t}.
\end{equation*}

\end{proof}

\begin{lemma}
\label{lem:proj}
Let $\Sigma_j = (I - \Pi_{j-1})\Sigma_1 (I - \Pi_{j-1})$, { where $\Pi_{j-1}=\sum_{k=1}^{j-1}\bv_k \bv_k^{\t}$}, and
$\Sigmahat_j = (I-\Pihat_{j-1})\Sigma_1(I-\Pihat_{j-1})$. Then,
\begin{align*}
     \|\Sigmahat_j - \Sigma_j\|_2&\leq 
     2\gamma_{j}\|\Pi_{j-1} - \Pihat_{j-1}\|_2 +  \gamma_1 \|\Pi_{j-1} - \Pihat_{j-1}\|_2^2.
\end{align*}
\end{lemma}
\begin{proof} 
To simplify the proof, denote $\Pi_{j-1} - \Pihat_{j-1}$ by $D$.
First,
\begin{align*}
\Sigmahat_j
=& (I-\Pi_{j-1})\Sigma_1(I-\Pi_{j-1})
+ D\Sigma_1(I-\Pi_{j-1})
+ (I - \Pi_{j-1})\Sigma_1 D
+ D\Sigma_1 D\\
=&\Sigma_j +  D\Sigma_j
+ \Sigma_j D
+ D \Sigma_1 D,
\end{align*}
{since $(I - \Pi_{j-1})\Sigma_1\Pi_{j-1}=0$.}

We have
$\|D\Sigma_j\|_2^2 = \tr(D\Sigma_j^2 D) \leq \|\Sigma_j\|_{op}^2 \|D\|_2^2 = \gamma_{j}^2 \|D\|_2^2$. 
Thus, $\|D\Sigma_j\|_2 \leq \gamma_{j} \|D\|_2$. Similarly, $\|\Sigma_j D\|_2 \leq \gamma_{j}\|D\|_2$. 

{ We have $\|D\Sigma_1 D\|_2^2 = \tr(D\Sigma_1 D^2 \Sigma_1 D) = \tr(\Sigma_1^{1/2} D^2\Sigma_1 D^2 \Sigma_1^{1/2})$.
Note that $\Sigma_1^{1/2} D^2\Sigma_1 D^2 \Sigma_1^{1/2}\preceq \gamma_1 \Sigma_1^{1/2} D^4 \Sigma_1^{1/2}$. Thus,
$\|D\Sigma_1 D\|_2^2 \leq \gamma_1 \tr(\Sigma_1^{1/2}D^4 \Sigma_1^{1/2}) = \gamma_1 \tr(D^2\Sigma_1 D^2)
\leq \gamma_1^2 \tr(D^4) = \gamma_1^2 \|D^2\|_2^2\leq \gamma_1^2 \|D\|_2^4$. 
We have thus got,
$\|D\Sigma_1 D\|_2 \leq \gamma_1 \|D\|_2^2$.}
Therefore,
\begin{align*}
    \|\Sigmahat_j - \Sigma_j\|_2
    \leq 2\gamma_{j}\|D\|_2 +  \gamma_1 \|D\|_2^2.
\end{align*}
The proof is now complete.
\end{proof}

Lemma \ref{lem:Pihat} below  confirms that the estimated eigenvectors are indeed mutually orthogonal.
\begin{lemma}
\label{lem:Pihat}
For $1\leq i, j\leq r$,
$\hat{\bv}_i^{\t}\hat{\bv}_j = \I(i=j)$.
\end{lemma}
\begin{proof}[Proof of Lemma \ref{lem:Pihat}]
Fix $j$ with $2\leq j\leq r$, we now prove
\begin{equation}\label{eq:vj}
  \hat{\bv}_j^{\intercal}\hat{\bv}_i=0,\,   \forall i=1,...,j-1.  
\end{equation}
Let $\gamma_1^{\ast}$ be the eigenvalue of $\hat{H}_j$ corresponding
to the leading eigenvector $\hat{\bv}_j$.
Then $\hat{H}_j \geq \gamma_1^{\ast}\hat{\bv}_j\hat{\bv}_j^{\t}$ and
$$
\langle \hat{\Pi}_{j-1}, \gamma_1^{\ast}\hat{\bv}_j\hat{\bv}_j^{\t}
\rangle
\leq 
\langle \hat{\Pi}_{j-1}, \hat{H}_j \rangle = 0.
$$
Thus,
$$0 = \langle \hat{\Pi}_{j-1}, \gamma_1^{\ast}\hat{\bv}_j\hat{\bv}_j^{\t}
\rangle = \gamma_1^{\ast}
\hat{\bv}_j^{\t}\hat{\Pi}_{j-1}\hat{\bv}_j.$$
It follows that \eqref{eq:vj} holds and the proof of the lemma is complete. 
\end{proof}

The following is Lemma C.3. in \cite{chen2015localized}.
\begin{lemma}\label{lem:Hv}
Let $\bu$ and $\bv$ be unit vectors of the same length. Then,
\begin{equation}
\frac{1}{\sqrt{2}}\| \bu - \bv\|_2
\leq \|\bu\bu^{\t}-\bv\bv^{\t}\|_2
\leq \sqrt{2}\|\bu-\bv\|_2.
\end{equation}
\end{lemma}

\begin{lemma}
\label{lem:vtilde}
Let $\Pi_{j-1}=\sum_{k=1}^{j-1}\bv_k \bv_k^{\t}$
and $\hat{\Pi}_{j-1}=\sum_{k=1}^{j-1}\bvhat_k \bvhat_k^{\t}$.
Define
$\tilde{\bv}_j=\frac{(I-\hat{\Pi}_{j-1})\bv_j}{\lVert I-\hat{\Pi}_{j-1})\bv_j \rVert_2}$
for $j=1,\dots, r$.
Then,
\begin{itemize}
     \item[(a).] $ \lVert\tilde{\bv}_j-\bv_j \rVert_2\leq 2 \lVert \hat{\Pi}_{j-1}-\Pi_{j-1}    \rVert_2$;
     \item[(b).] $\tilde{\bv}_j\tilde{\bv}_j^{\intercal}\in \mathcal{F}_{\hat{\Pi}_{j-1}}$;
     \item[(c).] $|J_1(\tilde{\bv}_j\tilde{\bv}_j^{\intercal})|\leq (|J_1(\bv_j)| + \sum_{k=1}^{j - 1} |J_1(\hat{\bv}_k)| )^2$;
     \item[(d).] $|J_2(\tilde{\bv}_j\tilde{\bv}_j^{\intercal})|\leq \left( |J_2 (\bv_j \bv_j^T)| + \sum_{k=1}^{j-1}|J_2(\hat{\bv}_k\hat{\bv}_k^T)|\right)^2$.
\end{itemize}
\end{lemma}
\begin{proof}[Proof of Lemma \ref{lem:vtilde}]
We first prove (a).
Note that for all $\bu,\bv\in \R^{p}$,
\begin{align*} 
\left\| \frac{\bu}{\lVert  \bu   \rVert_2} -\frac{\bv}{\lVert \bv   \rVert_2}   \right\| \leq 2\frac{\lVert \bu-\bv    \rVert_2}{\lVert  \bu   \rVert_2\vee \lVert \bv   \rVert_2},
\end{align*}
and also
\begin{align*}
\lVert (I-\hat{\Pi}_{j-1})\bv_j   \rVert_2
\leq \lVert I-\hat{\Pi}_{j-1}   \rVert_{op}\leq 1.
\end{align*}
Hence,
\begin{align*}
\lVert  \tilde{\bv}_j-\bv_j   \rVert_2&= \left\|  \frac{(I-\hat{\Pi}_{j-1})\bv_j}{\lVert (I-\hat{\Pi}_{j-1})\bv_j    \rVert_2}-\frac{\bv_j}{\lVert \bv_j \rVert_2}\right\| _2\nonumber\\
&\leq 2\frac{\lVert (I-\hat{\Pi}_{j-1})\bv_j-\bv_j  \rVert_2}{\|(I-\hat{\Pi}_{j-1})\bv_j \|_2\vee \|\bv_j \|_2}\nonumber\\
&= 2\lVert \hat{\Pi}_{j-1}\bv_j \rVert_2.
\end{align*}
Since $\Pi_{j-1}\bv_j=0$,
\begin{align*}
\lVert  \tilde{\bv}_j-\bv_j   \rVert_2 
&\leq
2\lVert \hat{\Pi}_{j-1}\bv_j \rVert_2\\
&=2\lVert (\hat{\Pi}_{j-1}-\Pi_{j-1})\bv_j \rVert_2\nonumber\\
&\leq 2\lVert (\hat{\Pi}_{j-1}-\Pi_{j-1})\bv_j   \rVert_2\nonumber\\
&\leq 2\lVert \hat{\Pi}_{j-1}-\Pi_{j-1}    \rVert_{op} \lVert \bv_j    \rVert_2\nonumber\\
&\leq  2\lVert \hat{\Pi}_{j-1}-\Pi_{j-1}    \rVert_2
\end{align*}
and we have proved (a).

We now prove (b). We just need to show that $\langle  \hat{\Pi}_{j-1}, \tilde{\bv}_j\tilde{\bv}_j^{\intercal} \rangle=0$, which follows by the fact that
$\hat{\Pi}_{j-1}$
is a projection matrix and by the definition of $\tilde{\bv}_j$.

We finally prove (c) and (d). By the definition of $\tilde{\bv}_j$ 
\begin{align*}
    \tilde{\bv}_j=\frac{(I-\hat{\Pi}_{j-1})\bv_j}{\|I-\hat{\Pi}_{j-1})\bv_j \|_2}=\frac{\bv_j-\sum_{i=1}^{j-1}(\hat{\bv}_i^{\t}\bv_j)\hat{\bv}_i}{\|(I-\hat{\Pi}_{j-1})\bv_j \|_2}.
\end{align*}
The proof is straightforward and hence is omitted.
\end{proof}

\begin{lemma}
\label{lem:Pi}
Let $\Pi_{j-1}=\sum_{k=1}^{j-1}\bv_k \bv_k^{\t}$
and $\hat{\Pi}_{j-1}=\sum_{k=1}^{j-1}\bvhat_k \bvhat_k^{\t}$.
Then,
$$\lVert \hat{\Pi}_{j-1}-\Pi_{j-1}    \rVert_2\leq \sqrt{2}\sum_{k=1}^{j-1} \|\bvhat_k - \bv_k\|_2.
$$
\end{lemma}
\begin{proof}
The proof is straightforward with Lemma \ref{lem:Hv}.
\end{proof}

\begin{lemma}
\label{df_ftp} (Deflated Fantope projector \cite{chen2015localized})
For any $p\times p$ symmetric matrix $A$ and projection matrix $\Pi$, the deflated Fantope projection operator is defined as $ \mathcal{P}_{\mathcal{D}_{\Pi}}(A) \coloneqq \arg\min_{B\in \mathcal{D}_{\Pi}} \|A-B \|_2^2$. Specifically, let $\Pi = VV^{\t}$, where $V$ is a $p \times d$ matrix with orthonormal columns and let $U$ be a $p\times (p-d)$ matrix that forms an orthogonal complement basis of $V$. Then 
\begin{equation*}
\mathcal{P}_{\mathcal{D}_{\Pi}}(A) = U\left[\sum_{i=1}^{p-d}\gamma_i^{+}(\theta)\eta_i\eta_i^{\t}\right]U^{\t},
\end{equation*}
where $(\gamma_i, \eta_i)_{i=1}^{p-d}$ are eigenvalue-eigenvector pairs of $U^{\t}AU$: $U^{\t}AU = \sum_{i=1}^{p-d}\gamma_i\eta_i\eta_i^{\t}$, and $\gamma_i^{+}(\theta)=\min(\max(\gamma_i-\theta,0),1)$, with $\theta$ chosen such that $\sum_{i=1}^{p-d}\gamma_i^{+}(\theta)=1$.
\end{lemma}

\begin{lemma}
\label{lem:curv} (Curvature Lemma \cite{vu2013fantope}) Let $A$ be a symmetric matrix and $E$ be the projection onto the subspace spanned by the eigenvectors of A corresponding to its $d$ largest eigenvalues $\gamma_1\geq \gamma_2\geq \ldots$. If $\delta_A = \gamma_d - \gamma_{d+1} > 0$, then
\begin{equation*}
\frac{\delta_A}{2}||| E - F |||_2^2 \leq \la A, E-F \ra
\end{equation*}
for all $F$ satisfying $0 \preceq F \preceq I$ and $\tr(F) = d$.
\end{lemma}

\begin{lemma}
\label{lem:ftp_prop}
(Basic properties of Fantope projection \cite{lei2015sparsistency})
Let $A$ be a symmetric matrix with eigenvalues $\gamma_1\geq \ldots \geq \gamma_p$ and orthonormal eigenvectors $\bv_1, \ldots, \bv_p$.
\begin{itemize}
    \item[1.] $\max_{H\in \mathcal{F}^d}\la A, H\ra = \gamma_1 + \ldots + \gamma_d$ and the maximum is achieved by the projector of a $d-$dimensional principal subspace of $A$. Moreover, the maximizer is unique if and only if $\gamma_d > \gamma_{d+1}$.
    \item[2.] $\mathcal{P}_{\mathcal{F}^d}(A) = \sum_j \gamma_j^{+}(\theta)\bv_j\bv_j^{\t}$, where $\gamma_j^{+}(\theta)=\min(\max(\gamma_j-\theta,0),1)$ and $\theta$ satisfies the equation $\sum_j\gamma_j^{+}(\theta)=d$.
    \item[3.] If $0<\tau\leq \gamma_d-\gamma_{d+1}$, then
    \begin{equation*}
\arg\max_{H\in\mathcal{F}^d} \la A, H \ra = \arg\max_{H\in \mathcal{F}^d}\la A, H\ra - \frac{\tau}{2} \| H\|_2^2 = \mathcal{P}_{\mathcal{F}^d} (\tau^{-1}A) = \sum_{j=1}^{d}\bv_j\bv_j^{\t},
    \end{equation*}
    uniquely.
\end{itemize}
\end{lemma}